\documentclass[twoside,leqno]{article}
\usepackage[letterpaper]{geometry}
\usepackage{ltexpprt}
\usepackage{hyperref}
\long\def\commentbegin #1\commentend{}
\usepackage{paralist}
\usepackage{mathrsfs}
\usepackage{amsmath}
\usepackage{amsfonts}
\usepackage{algorithm}
\usepackage[noend]{algpseudocode}
\usepackage{comment}
\usepackage{bbm}
\usepackage{thm-restate}
\usepackage{pdfpages}

\usepackage{graphicx, color, xcolor}
\usepackage{relsize}
\usepackage{times}
\usepackage{setspace}
\usepackage{float}
\usepackage{enumitem}
\usepackage{soul}
\usepackage[normalem]{ulem}
\usepackage{url}

\usepackage{tcolorbox} %For High-level flow boxes

\newcommand{\eps}{\varepsilon}

\newcommand{\E}{\text{E}}

\date{}

\newtheorem{definition}{Definition}
\newtheorem{claim}{Claim}

\def\polylog{\operatorname{polylog}}

\newcommand{\I}{{\mathbbm{1}}}

\def\capa{{\textsf{cap}}}
\def\tok{{\textsf{w}}}
\def\tot{{\textsf{total}}}
\def\rwLength{{\textsf{rwLength}}}
\def\numPhases{{\textsf{numPhases}}}
\def\phaseNum{{\textsf{phaseNum}}}
\def\stepNum{{\textsf{stepNum}}}

\def\outbox{{\textsf{Outbox}}}

\AtBeginDocument{%
  }

\title{Fully-Distributed Byzantine Agreement in Sparse Networks}

\author{John Augustine \thanks{Department of Computer Science and Engineering, Indian Institute of Technology at Madras, Chennai, Tamil Nadu, 600036, India. Email: \href{mailto: augustine@iitm.ac.in}{\tt augustine@iitm.ac.in}. He is also affiliated with and supported by the Centre for Cybersecurity, Trust and Reliability (CyStar) at IIT Madras.}  \and Fabien Dufoulon \thanks{School of Computing and Communications, Lancaster University, Lancaster, UK.  Email: \href{mailto: f.dufoulon@lancaster.ac.uk}{\tt f.dufoulon@lancaster.ac.uk}.}  \and Gopal Pandurangan \thanks{Department of Computer Science, University of Houston, Houston, TX 77204, USA. Email: \href{mailto: gopal@cs.uh.edu}{\tt gopal@cs.uh.edu}. Supported in part by ARO Grant W911NF-231-0191 and NSF grant CCF-2402837.}}

\begin{document}
\maketitle

Byzantine agreement is a fundamental problem in fault-tolerant distributed networks that has been studied intensively for the last four decades. Most of these works designed protocols for \emph{complete} networks.  A key goal in  Byzantine protocols is to tolerate as many Byzantine nodes as possible.

The work of Dwork, Peleg, Pippenger, and Upfal [STOC 1986, SICOMP 1988] was the first to address the  Byzantine agreement problem in \emph{sparse, bounded degree} networks and presented a protocol that achieved \emph{almost-everywhere} agreement among honest nodes.
In such networks, all known  Byzantine agreement protocols (e.g., Dwork, Peleg, Pippenger, and Upfal, STOC 1986; Upfal, PODC 1992; King, Saia, Sanwalani, and Vee, FOCS 2006) that tolerated a large number of Byzantine nodes had a major drawback that they were \emph{not fully-distributed} --- in those protocols, nodes are required to have initial knowledge of the entire network topology.  This drawback makes such protocols inapplicable to real-world communication networks such as peer-to-peer (P2P) networks, which are typically sparse and bounded degree and where nodes initially have only \emph{local} knowledge of themselves and their neighbors. Indeed, a fundamental open question raised by the above works is whether one can design Byzantine protocols that tolerate a large number of Byzantine nodes in sparse networks that work with only \emph{local}  knowledge, i.e., {\em fully-distributed} protocols.
The work of Augustine, Pandurangan, and Robinson [PODC 2013] presented the first fully-distributed Byzantine agreement protocol that works in sparse networks, but it tolerated only up to $O(\sqrt{n}/\polylog{n})$ Byzantine nodes (where $n$ is the total network size). 

We present   {\em fully-distributed} Byzantine agreement protocols for sparse, bounded degree networks that tolerate significantly more Byzantine nodes,  answering the earlier open question. Our protocols work under the powerful \emph{full information model} where the Byzantine nodes can behave arbitrarily and maliciously, have knowledge about the entire state of the network at every round, including random choices made by the nodes up to (and including) the current round,  have unlimited computational power, and may collude among themselves.  We first present a protocol that tolerates 
up to $o(\frac{n}{\log n})$ Byzantine nodes and with high probability\footnote{Throughout, ``with high probability (whp)'' means with probability at least $1 - \frac{1}{n^c}$, where $c \geq 1$ is a fixed constant.}, solves  {\em almost-everywhere} agreement where all except $o(n)$ honest nodes reach agreement. The protocol runs in $\tilde{O}(n^2)$ rounds. 
We then present a faster protocol that runs in nearly linear (i.e., $\tilde{O}(n)$) rounds and tolerates up to $o(\frac{n}{\log^2 n})$ Byzantine nodes. Both protocols are communication-efficient in the sense that honest nodes send only $\polylog{n}$ bits per edge per round.

\section{Introduction}
\label{sec:introduction}
Distributed computation in the presence of faulty and malicious nodes (also called \emph{Byzantine nodes}) is a central problem in distributed network algorithms.
 The Byzantine agreement problem can be stated as follows.
 \begin{definition}[Byzantine Agreement (BA)]
Let $P$ be a protocol on a distributed network of $n$ nodes in which each node $v$ starts with an input bit value $b_v$.
A Byzantine adversary controls up to $t$ nodes, which are called Byzantine (or faulty), which can deviate arbitrarily from $P$.
Protocol $P$ solves Byzantine agreement  if each (honest) node $v$ running $P$ terminates and  outputs a value $o_v$ at the end
of $P$ such that:
\begin{compactdesc}
\item [Agreement:] For any two honest nodes $u$ and $v$, $o_u = o_v$.
\item [Validity:] If the input value for all nodes is $b$, the  output value for all honest nodes
should be $b$.
\end{compactdesc}
\end{definition}

Much of the work in distributed Byzantine protocols has focused on \emph{complete networks}, starting from the classical work of  Pease, Shostak and Lamport \cite{Pease_1980} on \emph{Byzantine agreement} in the early 1980s and followed by several others in the last four decades, see e.g., \cite{Rabin_1983, Ben-Or_1983,  Feldman_1997, Goldwasser_2006, Ben-Or_2006, King_2006_SODA, King_2011, Augustine_2020_DISC}. However, as pointed out over 35 years ago in a seminal paper by Dwork, Peleg, Pippenger, and Upfal \cite{Dwork_1988}, such protocols on complete networks {\em do not} directly apply to {\em real-world networks}, which are typically {\em sparse and of bounded degree}.\footnote{In this work, by ``bounded degree'', we mean that the maximum degree is bounded by some fixed {\em constant}. However, our protocol and results can be easily extended to apply
if the degree is bounded by a slow-growing function of $n$, say, $\polylog{(n)}$, where $n$ is the number of nodes in the network.}  For example, the Bitcoin Peer-to-Peer (P2P) network allows eight outgoing connections and up to 125 incoming connections \cite{Mao_2020}.

The Dwork et al. \cite{Dwork_1988} paper was the first to study the fundamental Byzantine agreement problem in \emph{sparse, bounded degree} networks and presented the first Byzantine agreement protocol that achieved {\em almost-everywhere agreement} among honest nodes, where agreement is guaranteed on \emph{most} (but not all) honest nodes.
More formally, in the \emph{Almost-Everywhere (binary) Byzantine Agreement (AEBA)} problem, each node starts with binary input values (as in the Byzantine agreement defined above) and must output binary output values satisfying: (1) \emph{almost-everywhere agreement}, that is, all honest nodes except $O(t)$ of them output the same value\footnote{The protocol
in this paper allows slightly more honest nodes to disagree: it achieves agreement among all honest nodes except $o(n)$ of them for $t = o(n/\log n)$ or $t = o(n/\log^2 n)$ Byzantine nodes.}, and (2) \emph{validity}, that is, if all honest nodes have the same input value, then all honest nodes except $O(t)$ of them output that value.
In bounded degree networks, since an adversary can always isolate some number of honest nodes, almost-everywhere agreement is the best one can hope for. Furthermore, it must be pointed out that there is a somewhat stringent condition that the network has to have $\Omega(t)$ connectivity to reach everywhere agreement under $t$ Byzantine nodes \cite{Dolev_1981, Hadzilacos_1984}. 
Thus, for sparse networks, one has to relax the notion of \emph{everywhere agreement} to \emph{almost-everywhere agreement}.

The work of Dwork et al. \cite{Dwork_1988} showed how one can achieve almost-everywhere Byzantine agreement in a $d$-regular \emph{expander} graph ($d$ is a constant).  
Their protocol can tolerate up to $O(n^{1 - \eps})$ Byzantine nodes (throughout $n$ denotes the total
number of nodes in the network), for some small constant $\eps$ (that depends on $d$) 
and achieves agreement among all but $O(n^{1-\eps})$ honest nodes. 
They also show that by superimposing a butterfly network with an expander graph, they can tolerate up to $O(\frac{n}{\log{n}})$ Byzantine nodes and achieve agreement among all but $O(\frac{n}{\log{n}})$ honest nodes. 
This work was improved in subsequent papers \cite{Berman_1993_DC, Berman_1993_MST, Upfal_1994}. In particular, Upfal \cite{Upfal_1994} presented an almost-everywhere Byzantine agreement protocol that can tolerate up to linear, i.e., $\epsilon n$ (for a sufficiently small $\epsilon > 0$), Byzantine nodes in a constant-degree expander network.

A major drawback of  the above protocols is that they  require \emph{initial knowledge of the global topology}\footnote{Also in some cases, specific network designs such as butterfly are assumed which are infeasible in large-scale decentralized networks.}, since at the start, nodes need to have this information ``hardcoded.''
Another drawback of the above results \cite{Dwork_1988, Upfal_1994, Berman_1993_MST, Berman_1993_DC} is that they require each node to use $\Omega(n^2)$ number of bits in communication. Thus, in a sparse network and in the CONGEST model of communication, where only $\polylog{(n)}$ bits of communication are allowed per edge per round \cite{Peleg_2000_Book}, these protocols take $\tilde{\Omega}(n^2)$ time.   
An additional drawback of Upfal's protocol \cite{Upfal_1994} is that the local computation required by each processor is exponential.

The work of  King, Saia, Sanwalani, and Vee \cite{King_2006_FOCS} was the first to study \emph{scalable} algorithms for Byzantine almost-everywhere agreement and leader election in sparse networks. By ``scalable'', it is meant that the total number of bits that any honest node communicates and processes throughout the protocol is at most $\polylog(n)$. This protocol tolerated
up to $ (\frac{1}{3} - \epsilon)n$ Byzantine nodes (where $\epsilon > 0$ is any constant) and achieved agreement among $1 - O(1/\ln n)$ fraction of honest nodes.
This work also requires that the nodes have {\em global topological knowledge}, i.e., hardcoded information on the network topology to begin with. Thus, the protocol of Dwork et al.\ and several subsequent protocols including that of King et al.\ {\em had the major drawback that they are  \emph{not fully-distributed} in the sense that nodes are required to have initial knowledge of the global network topology.} This drawback makes such protocols not applicable to real-world communication networks such as Peer-to-Peer (P2P) networks, which are typically sparse and dynamic, with nodes having only \emph{local} initial knowledge of only themselves and their neighbors. 
A fundamental question left open by the works of \cite{King_2006_FOCS, Upfal_1994, Dwork_1988} is whether one can design Byzantine protocols in sparse networks that work with only \emph{local}  knowledge, i.e., {\em fully-distributed} protocols.
In fact, Dwork et al.\ \cite{Dwork_1988} pose the following  open question in their paper (a similar question
is also posed in King et al. \cite{King_2006_FOCS}):
\begin{center}
    \fbox
    {
        \begin{minipage}{35em}
            ``The algorithms presented in this work (i.e., \cite{Dwork_1988}) assume that all the processors know the topology of the communication network and the communication schemes used by all other processors. Is this requirement essential for achieving almost-everywhere agreement?''
        \end{minipage}
    }
\end{center}

The work of Augustine, Pandurangan, and Robinson \cite{Augustine_2013_PODC} gave the first known fully-distributed Byzantine agreement protocol 
that worked in sparse networks.\footnote{In fact, their protocol can also handle churn and work in dynamic networks.} However, a drawback of this protocol is that it could tolerate only $\sqrt{n}/\polylog{n}$ Byzantine nodes.

In this paper, we  present  {\em fully-distributed} Byzantine agreement protocols for  sparse bounded-degree networks
that tolerate {\em significantly more Byzantine nodes},  up to $o(\frac{n}{\log n})$  Byzantine nodes.

\subsection{Model}
\label{sec:model}

\noindent {\bf Network Model.} We assume an arbitrary $d$-regular {expander}  graph $G$, where $d$ is a constant.\footnote{Our protocol can be easily extended to work with non-constant $d$ (say, $\polylog{n}$) as well. The regularity assumption can also be relaxed if the degrees are within a small (say, constant or $\polylog{n}$) factor of each other.} Expander graphs \cite{Hoory_2006} have conductance at least a constant (independent of $n$, the network size).
We note that this is the same graph model used in the work of Dwork et al.\ \cite{Dwork_1988} discussed earlier.\footnote{Dwork et al.\ \cite{Dwork_1988} actually assume a $d$-regular random graph which is an expander with high probability. Furthermore, since their result holds with high probability on a random $d$-regular graph, they claim that their result holds for ``almost all'' $d$-regular graphs.  We use the same model in this paper, and hence, the same remark applies to our result as well, although we tolerate a higher number of Byzantine nodes.}
Note that the above graph model is quite general in the sense that we only assume that the topology is an expander graph; no other special properties are assumed. Indeed, expander graphs have been used extensively as candidates to solve the Byzantine agreement and related problems in sparse graphs in prior works \cite{Dwork_1988, Kapron_2010, King_2011, King_2006_SODA, Upfal_1994}; the high expansion is crucial in these results, in particular, for tolerating a large number of Byzantine nodes. Expander graphs have also been used extensively to model P2P networks\footnote{In particular, the real-world  Bitcoin P2P network, constructed by allowing each node to choose eight random (outgoing) connections (\cite{Mao_2020}) is likely an expander network if the connections are chosen (reasonably) uniformly at random \cite{Palmer_1985_Book}.} (see e.g., \cite{Law_2003,  Pandurangan_2003, Jacobs_2013, Cooper_2007, Mahlmann_2006, Augustine_2012, Augustine_2013_PODC, Augustine_2015_DISC, Augustine_2013_SPAA,sigact16}).

\medskip

\noindent {\bf Local Knowledge.} An important assumption in sparse networks is that nodes at the beginning have only \emph{local} knowledge, i.e., they have knowledge of only themselves and their neighbors in $G$. In particular, they do \emph{not} know the global topology or the identities of other nodes (except those of their neighbors) in the network. Thus, we seek {\em fully-distributed} protocols where nodes start with only local knowledge. However, as is common in distributed computing literature (see e.g., \cite{Peleg_2000_Book,Kutten_2015_TCS}),  we  do assume that nodes have knowledge of  global parameters such as the network size $n$ (or a good, say constant-factor, approximation of it) as well as the conductance of the expander, which  
is a constant. (A constant lower bound of the conductance is enough --- this gives an upper bound on the mixing time
which is $O(\log n)$ that is used in the protocol.)

\medskip

\noindent {\bf Full Information Model.} We assume the powerful {\em full-information} model (e.g., see \cite{Ben-Or_2006, coin-flipping,linial-full-info}) that has been studied extensively.  In this model,  the Byzantine nodes (controlled by an adversary) can behave arbitrarily and maliciously and have knowledge about the entire state of the network at every round, including random choices made by all the nodes up to and including the current round (this is also called {\em rushing} adversary), have unlimited computational power, and may collude among themselves (hence, cryptographic techniques are not applicable in this setting).  We assume a {\em static} Byzantine adversary where the nodes
that are corrupted are chosen {\em before} the start of the protocol. 

\medskip

\noindent {\bf Communication Model.} Communication is {\em synchronous} and occurs via  message passing, i.e., communication proceeds in discrete rounds
by exchanging messages on {\em the edges} of $G$, i.e., each node (including Byzantine nodes) can exchange messages
 only with its neighbors in $G$. 
By our protocol design, honest nodes will only send $O(\polylog n)$ bits per edge per round. Note that Byzantine
nodes do not have any such limit and can send as many bits as they want. (Our protocol is designed in such a way
that it crucially handles this extra power given to Byzantine nodes without limiting in any way the bandwidth capacity
of the edges.)
 As is standard in Byzantine algorithms (see, e.g., Lamport et al. \cite{Pease_1980}), we assume that the receiver of a message across an edge in $G$ knows the identity of the sender, i.e., if $u$ sends a message to $v$ across edge $(u,v)$, then $v$ knows the identity of $u$; also the message sent across an edge is delivered correctly and in order.

\subsection{Our Contributions} \label{sec:result}

Our main contributions are   {\em fully-distributed} Byzantine agreement protocols for sparse, bounded-degree networks  that  tolerate a large number of Byzantine nodes in the full information model. We first present a protocol that tolerates 
up to $o(\frac{n}{\log n})$ Byzantine nodes (where  $n$ is the total number of nodes). 
The protocol, with high probability\footnote{Throughout, ``with high probability (whp)'' means with probability at least $1 - \frac{1}{n^c}$, where $c \geq 1$ is a fixed constant.}, solves  {\em almost-everywhere} agreement where all except $o(n)$ honest nodes reach agreement. The protocol requires honest nodes to send only $\polylog{n}$ bits per edge per round. The protocol takes  $\tilde{O}(n^2)$ rounds. 
We then present a faster (somewhat more complicated) protocol that runs in nearly linear (i.e., $\tilde{O}(n)$) rounds and tolerates slightly less Byzantine nodes --- up to $o(\frac{n}{\log^2 n})$ Byzantine nodes.

Our protocols are a significant improvement over previous Byzantine protocols in sparse (expander) networks \cite{Dwork_1988, Upfal_1994, King_2006_FOCS} in that it {\em requires only local (initial) knowledge} and answers the open question raised in Dwork et al. \cite{Dwork_1988} of whether such a protocol is possible. Our first protocol's runtime is comparable
 to the protocols of Dwork et al. and Upfal \cite{Upfal_1994} --- as these also take
 $\tilde{\Omega}(n^2)$ rounds under $\polylog n$ bandwidth constraint. Our second protocol's runtime of $\tilde{O}(n)$ rounds is significantly 
 faster while tolerating up to $o(\frac{n}{\log^2 n})$ Byzantine nodes. We note that the protocol of Dwork et al. 
 tolerates  $O(n^{1 - \eps})$ Byzantine nodes  and the protocol of Upfal  tolerates $\eps n$ nodes  (cf. Section \ref{sec:introduction}).
 
 The protocol of King et al. \cite{King_2006_FOCS} is much faster,
 taking $O(\polylog{n})$ rounds\footnote{Although the exact power of $\log n$ is not explicitly specified in the paper,
 it is a somewhat large constant.} and tolerates up to $ (\frac{1}{3} - \epsilon)n$ Byzantine nodes, but, as pointed out earlier, this protocol needs {\em initial knowledge of the global expander topology}  and, hence is
 {\em not fully-distributed}, unlike our protocol. As mentioned in Section \ref{sec:introduction}, being fully-distributed is crucial for implementation in real-world networks such as P2P networks, which are sparse and bounded-degree with nodes
 having only local initial knowledge of themselves and their neighbors.
 
 Our protocols  improve upon the fault-tolerance of the Byzantine agreement protocol
 of Augustine et al.~\cite{Augustine_2013_PODC} that, while being fully-distributed and fast (running in $O(\log^3 n)$ rounds),
 could tolerate only $\sqrt{n}/\polylog{n}$ Byzantine nodes. {\em Designing a fully-distributed agreement
 protocol that tolerates significantly more Byzantine nodes (say, up to $O(n/\polylog{n}$)) while keeping an extremely fast runtime (say, in $\polylog{n}$ rounds) is an important open problem} (cf. Section \ref{sec:conclusion}).

Our tools and techniques (discussed in the next section) are likely to be of independent interest and 
can be useful for designing fully-distributed Byzantine protocols for other important problems, such as leader and committee elections, in the more challenging setting of sparse networks. In particular, we define and present two primitives, namely {\em Almost-Everywhere Reliable Information Dissemination (AERID)} (cf. Section \ref{sec:AERID}) and {\em Eventual Almost-Everywhere  Common Coin (EAECC)} (cf. Section \ref{sec:randomCoin}). The former allows most honest nodes to correctly disseminate information to most other honest nodes, even under the presence of a large number of Byzantine nodes, in a fully-distributed fashion. Our AERID protocols both take $\tilde{O}(n)$ rounds, which is (essentially) optimal in a bounded degree network. We then show how to use AERID to design a protocol to implement an eventual almost-everywhere common coin (EAECC), allowing most honest nodes to agree on a random common coin value eventually. We give two protocols for EAECC, one for each AERID protocol. The first EAECC protocol takes $\tilde{O}(n^2)$ rounds under the presence of $o(n/\log{n})$ Byzantine nodes, and the second takes $\tilde{O}(n)$ rounds under the presence of $o(n/\log^2{n})$ Byzantine nodes.

\subsection{Prior Work and Challenges} \, \smallskip

\noindent{\bf Prior Work.}
At a high level, we take a different approach to designing Byzantine protocols compared to previous works, including Dwork et al.\ \cite{Dwork_1988}, Upfal \cite{Upfal_1994}, and King et al.\ \cite{King_2006_FOCS}.
Unlike these prior protocols,  our protocol is  \emph{fully-distributed} (i.e., works with only local initial knowledge). To see the contrast, we briefly discuss the protocol scheme of Dwork et al. for the Byzantine agreement. (A similar scheme is also used in \cite{Upfal_1994, Berman_1993_MST, Berman_1993_DC}). The main idea of this scheme is to simulate a Byzantine agreement protocol designed for a {\em complete network} in a sparse network $G$. This simulation is done by specifying a transmission scheme that ensures reliable communication between most honest processors; this communication is accomplished by sending a message over multiple paths in $G$ between two nodes (the paths are different for different pairs of nodes).
For the transmission scheme to work, each (or most) honest node(s) should know their respective communication paths to other nodes. This requires that all honest nodes know the topology of $G$. 

King et al.~\cite{King_2006_FOCS} design a (specific) expander network with a polylogarithmic degree whose topology needs to be known by all honest nodes for efficient and reliable communication. They raise the question of whether one can design a scalable and robust protocol that works correctly on {\em any} sparse network with sufficiently good expansion.
Our result gives such protocols that work in {\em any} sparse expander graph.

 Augustine et al.~\cite{Augustine_2013_PODC} presented the first known fully-distributed Byzantine agreement protocol
 for a sparse constant-degree (regular) expander network. Their protocol also worked under a dynamic churn setting.
 The protocol is based on implementing the following random sampling algorithm which is first stated
 for a complete network. In each round, each node samples two random nodes uniformly
 at random in the network and takes the {\em majority}  of its value and the two sampled values.
It can be shown \cite{Augustine_2013_PODC,two-choices} that this protocol converges to a common value in $O(\log n)$ rounds provided that
the number of Byzantine nodes is $O(\sqrt{n})$. The protocol's advantage is that the total number of random samples
requested by any honest node is $O(\log n)$.
The novelty of the protocol of Augustine et al. is to implement the above randomized sampling algorithm
efficiently in a sparse constant degree expander. The main technical tool used is {\em random walks}, which is
also used in the current paper (though there are significant differences in how and why they are used, as discussed
below). 
They present a result called the Byzantine Sampling Theorem that characterizes random walks on expander 
networks in the presence of a large number (up to $O(\sqrt{n}/\polylog{n})$ Byzantine nodes). It shows
that despite the presence of a large number of Byzantine nodes, most random walks mix rapidly and reach the uniform stationary
distribution. This allows most honest nodes to sample two honest nodes almost uniformly randomly.
Using this result, one can implement the majority agreement protocol described earlier in a sparse network
and show that it converges in $\polylog{n}$ rounds. It is crucial to note that the majority agreement fails
if the number of Byzantine nodes is somewhat larger, say,  $n^{1/2 +\epsilon}$ for any small constant $\epsilon$.
In this case, the agreement protocol can take {\em exponential} (in $n$) number of rounds.
Thus, tolerating a much larger number of Byzantine nodes via a fully-distributed protocol in reasonable (say, polynomial number of rounds) time remained an open problem.

\smallskip

\noindent{\bf The Advantage of Random Walks.}
Before we give a high-level idea of our approach, we give
an intuition as to why random walks work well in a sparse network  (unlike broadcast, for example).  Random walks
are lightweight (and local) and allow us to bound the number of messages sent by {\em Byzantine nodes}. Byzantine nodes need not follow
the random walk protocol and can send a lot of messages, but once these messages reach honest nodes, their influence
becomes limited. We give a protocol called the Byzantine Random Walk Protocol (cf. Algorithm \ref{alg:byzantineSamplingSparse}) and prove a key result called  ``The Byzantine Random Walk Theorem'' (cf. Theorem \ref{thm:ByzSamplingFinal}) that shows how the protocol can implement random walks in bounded-degree networks with
a large number of Byzantine nodes.  This theorem shows precisely how the Byzantine Random Walk Protocol controls the messages sent by Byzantine nodes
and how their influence can be limited for most random walks initiated by most honest nodes.
Though the Byzantine Random Walk Theorem (Theorem \ref{thm:ByzSamplingFinal}) of the current paper is similar in spirit to the Byzantine Sampling Theorem of \cite{Augustine_2022_SPAA}, it is stronger in the sense that it explicitly specifies the subgraph of $G$, called the {\em core} graph which is itself an expander and establishes the mixing of (most of) the random walks on this core graph. A big advantage of
the Byzantine Random Walk Protocol is that it can be used to achieve {\em Almost-Everywhere  Reliable Information Dissemination (AERID)} (see Section \ref{sec:AERID}) even
under a large number of Byzantine nodes. This guarantees that most honest nodes will be able to broadcast the data item they possess correctly to most other honest nodes in near-linear, i.e., $\tilde{O}(n)$, rounds, which is optimal in a constant degree network.\footnote{\label{ft:aerid} $\Omega(n)$ rounds are needed for AERID, since each node from a set of $n-o(n)$ nodes has to receive information from  $n-o(n)$ nodes, and since the degree is bounded, any honest node can send/receive only $\polylog{n}$ bits of information in a round.}

We note that random walk techniques inherently cannot tolerate $\Omega(n/\log n)$ Byzantine nodes, since the mixing time needed in a sparse expander network is at least $\Theta(\log n)$.  
If there is a linear number of Byzantine nodes, then most random walks will go through a Byzantine node, and the Byzantine Random Walk Theorem does not work. We conjecture that no {\em fully-distributed} algorithm may tolerate $\omega(n/\log n)$ Byzantine nodes in a sparse network.
Our first protocol reaches close to this limit, i.e., it can tolerate up to $o(n/\log n)$ Byzantine nodes.

\smallskip

\noindent{\bf Almost-Everywhere Relaxations.}
A basic difficulty that we have to overcome when the communication graph is a sparse ($d$-regular) {\em expander}, and there exists a Byzantine adversary controlling some $t \geq d$ nodes, is that the Byzantine adversary can disconnect up to $O(t/d)$ honest nodes. As a result, many basic distributed primitives (such as leader election, broadcast, spanning tree construction, etc.) are impossible to solve in their traditional formulations. Instead, we aim for \emph{almost-everywhere} relaxations of these problems, which build on the following fact. For the Byzantine adversary, controlling up to $O(n / \polylog n)$ nodes, the adversary can partition the honest nodes into a large {\em ``core''} subset of  $n-o(n)$ (honest) nodes and many possibly disconnected small subgraphs of up to $o(n)$ nodes each. Because
of the expander nature of the communication graph, it can be shown that there exists such a core subset of size $n-o(n)$ nodes, which is
a giant component, but also, most importantly, it is an expander.
As a result, when we adapt the traditional formulation to an almost-everywhere formulation, we mean solving the traditional problem only on the large core of honest nodes. The Byzantine Random Walk Theorem also exploits this to show
that the random walks mix well in a large core subset.

\subsection{A High-level Overview of our Approach}
We show that in sparse bounded-degree networks, Almost-Everywhere Byzantine Agreement (AEBA) can be solved by  
fully-distributed protocols in the full information model. We present two protocols, both of which follow the same approach.  The first tolerates up to $o(n/\log n)$ Byzantine nodes and runs in $\tilde{O}(n^2)$ rounds, while the second tolerates slightly less number of Byzantine nodes --- $o(n/\log^2 n)$ --- but runs faster in $\tilde{O}(n)$ rounds. 

\smallskip

\noindent{\bf Rabin's Scheme.} The high-level idea of our approach is conceptually simple and is based
on the well-known randomized agreement scheme of Rabin \cite{Rabin_1983} that has been used
to obtain efficient protocols (running in $\polylog{n}$ rounds) in {\em complete networks} (see e.g.,
\cite{Feldman_1997, Ben-Or_2006, Goldwasser_2006}). However, it is non-trivial to implement this scheme in a {\em sparse (bounded-degree)}
network, which is what this work accomplishes. The basic idea of this scheme, described for
a complete network and modified to flow better into the exposition below (in particular, the original scheme does not need a strong majority), is as follows. For convenience, we assume that the number of Byzantine nodes $t$ is $o(n)$ (although
in a complete network, it works for up to $n/3$ Byzantine nodes in the full information model). 
Each node $v$ holds a binary \emph{decision variable} $b_v \in \{0,1\}$, whose initial value is given as input. 
The protocol operates in phases.
In each phase, node $v$ checks whether a \emph{strong majority} (say, at least a $0.9$ fraction) exists, i.e., whether almost all nodes have the same value in $\{0,1\}$. (This is done by having all nodes broadcast their $b_v$ values.) The number of nodes having the majority is called the {\em tally} of the majority value; for a strong majority, the tally should be at least $0.9n$. If such a value exists, then $v$ changes its decision variable to that value. (Note that there can be only one such value due to the majority.) 
Otherwise, it sets its value to the value given by a {\em global (shared)} common coin. (A common coin takes both 1 and 0 with constant probability.) It can be shown that this protocol converges in $O(\log n)$ phases with high 
probability. The main challenge in Rabin's protocol is implementing the global common coin efficiently. Rabin's paper
assumed that the common coin is given to all nodes by a trusted third party. The works of
\cite{Ben-Or_2006, Goldwasser_2006} showed how a common coin can be generated efficiently (in $O(\log n)$ rounds) by the honest nodes themselves in a {\em complete network} in the full-information model.

\smallskip

\noindent{\bf Implementing Rabin's Scheme in a Sparse Network.}
There are two main challenges in efficiently implementing Rabin's scheme in a {\em sparse} network in a fully-distributed manner.
First, how do nodes compute whether a strong majority exists in a sparse network? Indeed, unlike in a complete network, a node's neighbors are only a small fraction of the network. Second, how does one implement a global common coin (without Byzantine agreement) in a sparse network in the full information model? (Note that cryptographic tools cannot be used in this setting.) 

We detect a strong majority (say, at least a $0.9$ fraction) by doing random sampling. Each honest node samples $\polylog{n}$ random nodes by performing random walks. The random walks are implemented by following the Byzantine random walk protocol. The protocol guarantees
that most honest random walks mix in the core subgraph $C$ and sample close to the stationary distribution of $C$.
Since the core is of size $n-o(n)$, there is (essentially) a strong majority among the core subset, which will be detected by most ($n-o(n)$) of the honest nodes in the core. More importantly,  most honest nodes in the core will never disagree on the strong majority value.
Some may detect that strong majority, but others may not. Then, the latter honest nodes will use the common coin flip. Rabin's protocol ensures that after $O(\log n)$ common coin flips, all honest nodes will converge to AEBA. 

It is important to note that in Rabin's protocol, the global random coins must be generated {\em after} the tally is computed in every iteration. Otherwise, the Byzantine adversary (even a static one) can ensure that Rabin's protocol does not terminate by choosing a majority and tally value based on the coin value.

Indeed, if the global random coins are generated up-front (before any iteration starts), then the Byzantine adversary (even a static one) can ensure that (almost-everywhere) agreement is not reached even if the adversary does not influence any of the global random coins. More concretely, consider the following scenario in a $n$-node clique. (Similar scenarios can be obtained also in sparse networks, assuming nodes sample to detect the majority value, but this would contribute an additional layer of complexity to the argument.) The adversary decides on an initial configuration such that one of the values is held by exactly a 0.9 fraction of the (honest) nodes (i.e., the tally threshold). The random coin value chosen (up-front) for the first iteration is the (initial) minority value with a probability of at least 1/2. Assuming that event happens, the adversary can ensure that no iteration reaches (almost-everywhere) agreement. (Recall that in the clique, all (honest) nodes broadcast their values.) To do so, in the first iteration, the adversary has the Byzantine nodes act as follows (where a Byzantine node may send different values to different neighbors). If the coin value does not change in the next iteration, then the Byzantine nodes send different values to the honest nodes such that exactly a 0.9 fraction of the honest nodes receive more than a 0.9 fraction of the current majority value (among all $n$ values), and thus exactly a 0.9 fraction of the honest nodes hold the current majority value in the next iteration. Otherwise, the byzantine nodes send different values to the honest nodes such that exactly a 0.9 fraction of the honest nodes see strictly less than a 0.9 fraction of the current majority value (among all $n$ values), and thus exactly a 0.9 fraction of the honest nodes hold the current iteration's minority value in the next iteration. The adversary repeats this process in the subsequent iterations. Moreover, the adversary can execute this process with only $O(1)$ Byzantine nodes.

\smallskip

\noindent{\bf  Almost-Everywhere Reliable Information Dissemination (AERID).} Implementing a global common coin is more involved. It uses a key primitive called Almost-Everywhere Reliable Information Dissemination (cf. Section \ref{sec:AERID}).
We show that AERID can be done in $\tilde{O}(n)$ rounds, which is essentially optimal (cf. Footnote \ref{ft:aerid}). As mentioned earlier, AERID is implemented by using the Byzantine Random Walk Protocol (cf. Algorithm \ref{alg:byzantineSamplingSparse}).
The two protocols differ in how the walks are implemented. In the first protocol, each node initiates $\tilde{O}(n)$ random walks (each token contains the source ID and other data from the source node, say, its bit value).
As per the Byzantine Random Walk Protocol, these $\tilde{O}(n)$  random walks are sent in batches (phases) of $\Theta(\log^3 n)$ tokens each; each token walks for $\Theta(\log n)$ steps in a phase. The Byzantine Random Walk Protocol ensures that
in each phase, almost all the tokens from almost all the nodes reach their destinations despite the actions of the Byzantine nodes. The following intuition helps us understand how the Byzantine nodes' actions are contained.
Consider the core subgraph $C$ of $G$ consisting of only honest nodes. The core is a giant component of size $n-o(n)$ and is an expander itself. Hence, a walk started in the core mixes fast (in $\Theta(\log n)$ rounds) if the walk stays in the core.
It is easy to show that most of the walks started in the core nodes walk only in $C$ and thus mix well in the core. 
What about the tokens sent by Byzantine nodes? After all, they may not follow the protocol and can send arbitrarily many tokens. However, the behavior of Byzantine nodes is limited, thanks to the good nodes at the core boundary (i.e., core nodes with Byzantine neighbors), which act as ``guards'' and effectively control the rate of tokens
sent by Byzantine nodes. Nevertheless, the Byzantine nodes can try to fake the tokens of some good nodes, but only
$o(n)$ of them are affected. Since each phase takes only $O(\log n)$ steps, it can be shown (cf. proof of Theorem
\ref{thm:ByzSamplingFinal}) that the protocol can tolerate up to $o(n/\log n)$ Byzantine nodes; the high-level intuition is that the bad tokens can be limited to $o(n/\log n) \cdot O(\log n) = o(n)$.

The second protocol is a bit more complicated and implements AERID differently. This protocol works in $\Theta(\log n)$ stages.
In the first stage, each node initiates $\Theta(\log^{3+2\varepsilon} n)$ tokens for some arbitrarily small constant $\varepsilon > 0$ (where the token contains the source ID and other data). These tokens do random walks for $\Theta(\log n)$ (or core mixing time) steps. 
In each successive stage, each node that was a destination for a token from a particular source initiates a new random walk token per received token. Thus, the number of tokens from each source doubles in each stage. By a careful load balancing argument, it can be shown that at the end of each stage the tokens from each source node are distributed almost uniformly at random among the (destination) nodes.
This is crucial to bind the effect of Byzantine nodes, which can send too many tokens.  We show that in $O(\log n)$ stages,
most of the honest nodes' data is disseminated to (almost) all nodes. And intuitively, since tokens walk $O(\log n)$ steps over $O(\log n)$
stages, the number of Byzantine nodes tolerated is $o(n/\log^2 n)$.
Both protocols implement AERID in $\tilde{O}(n)$ rounds, but the advantage of the second protocol (though it tolerates slightly fewer Byzantine nodes) is that the broadcast time from a single source node is $\tilde{O}(1)$ rounds in contrast to the $\tilde{O}(n)$ rounds for the first protocol. This advantage can be used to implement a common coin much faster than the first, as outlined below.

\smallskip

\noindent {\bf Eventual Almost-Everywhere Common Coin (EAECC).} A main novelty of the protocols lies in the implementation of a (weak variant of) a common coin in a sparse network with up to (essentially) $n/\log{n}$ Byzantine nodes.
Instead of a common coin, we define and use what is called an \emph{Eventual Almost-Everywhere Common Coin (EAECC)}. 
The high-level idea is as follows. Each node, in turn, chooses to generate a random coin value (0 or 1 with equal probability) and broadcasts it to the network. The turn is decided by the rank of the node: an integer chosen uniformly at random in $[1,n]$ by each honest node.\footnote{Throughout, when talking about integers, we use $i \in [1,n]$ to mean $i \in \{1,\ldots,n\}$.}
The main challenge in implementing this idea lies in ensuring an honest node's turn (or coin flip) is not corrupted by Byzantine nodes. Indeed, the Byzantine nodes may try to broadcast arbitrary values and confuse a large portion of the honest nodes about the turn's random coin value. We handle this in the following way by crucially using AERID as a preprocessing step.

In both versions of the AERID protocols (cf. Section \ref{sec:AERID}),  we use the Byzantine Random Walk Theorem to prove that when honest nodes start sufficiently many random walks, most of these walks are completely unaffected by the Byzantine nodes. In both versions, the honest nodes compute these random walks {\em simultaneously as a preprocessing step} where the paths taken by the tokens are recorded
by the nodes. More precisely, each node records the source ID
of the token it receives, the token ID (identifying that token among all others with the same source ID), the random walk step number (number of steps taken by the random walk till now) when it receives the token, the incoming edge through which the token arrived and the forwarding edge of the token (if the current node is not the destination). Thus, effectively, the paths taken by each token from the source to the destination are recorded by the nodes through which the token is traversed.

These recorded paths are crucially {\em reused} to broadcast a common coin value from a particular source (containing the source ID and the random bit). During the broadcast, in each round, only tokens that conform to the recorded paths from this particular source 
node for this round are allowed. Honest nodes at the boundary of the core will enforce the recorded paths and ignore all other messages that do not correspond to these recorded random walks, which is the key to strongly limiting the negative impact Byzantine nodes can have on most broadcasts. 

We note that we precompute the random walks so that disseminating the common coin can be done efficiently. 
Instead, one could compute random walks in each iteration to avoid doing it upfront. However, in each iteration, we must compute random walks starting from many nodes (almost all honest nodes) to deal with the high number of Byzantine nodes, leading to strong inefficiencies. Hence, we precompute all random walks and do so in a special way to reduce congestion.

The above setup allows us to implement an EAECC. 
Indeed, recall that Rabin's protocol succeeds with high probability after $O(\log n)$ random coin flips. But in sparse networks, even in the worst case, Byzantine nodes can target only the common coin flips of the first $n/\log{n}$ ranked honest nodes.
Eventually, there are enough uncorrupted common coin flips that the honest nodes (almost-everywhere) agree, or in other words, the protocol converges.

\subsection{Additional Related Work} \label{sec:related}

The literature on Byzantine agreement is vast (especially on complete networks), and we limit ourselves to those that are most relevant to this work, mainly
focusing on sparse networks.

Most prior works on Byzantine protocols on sparse networks {\em assume} an underlying {\em expander} graph, where the expansion properties prove crucial in solving fundamental problems such as agreement and leader election, see, e.g., \cite{Dwork_1988, Upfal_1994, King_2006_FOCS}.  The protocol of \cite{King_2006_FOCS} builds an underlying communication mechanism where messages can be relayed with only $\polylog{(n)}$ overhead. The issue with all the above protocols, as mentioned earlier, is that they assume that nodes have global knowledge of the network topology to begin with. Such an assumption does not work where nodes start with local knowledge of only themselves and their immediate neighbors, as is common in real-world P2P networks (including those that implement cryptocurrencies and blockchains), which are bounded degree and sparse. 

Berman and Garay \cite{Berman_1993_DC,Berman_1993_MST} improved on the efficiency of Dwork et al \cite{Dwork_1988}. Their main result is an algorithm that achieves consensus in the butterfly network using $O(t + \ln n \ln\ln n)$ one-bit parallel transmission steps while tolerating $t = O(n/ \ln n)$ corrupted processors and having $O(t \ln t)$ confused processors (i.e., uncorrupted processors that have decided on the incorrect bit). The number of rounds, corrupted processors that can be tolerated, and confused processors in this result are all asymptotically optimal for the butterfly network.  
Ben-Or and Ron designed a bounded degree network and an almost-everywhere agreement algorithm that is fully polynomial and tolerates a linear number of faults with high probability if the faulty processors are randomly located throughout the network \cite{Ben-OrR96}.
King et al. \cite{King_2006_SODA} describe protocols for Leader Election and Byzantine Agreement that take polylogarithmic rounds and require each processor to send and process a polylogarithmic number of bits. These protocols only run on {\em complete} networks and do not apply to sparse networks.  

The  work of \cite{Augustine_2013_PODC}   presented a fully-distributed algorithm for  Byzantine agreement 
in the presence of Byzantine nodes and high adversarial churn. The algorithm could tolerate (only) up to  
$\sqrt{n}/\polylog{(n)}$ Byzantine nodes  and up to $\sqrt{n}/\polylog{(n)}$ churn per round and took a $O(\polylog{(n)})$ number of rounds.  
The work of \cite{Augustine_2015_DISC} used the Byzantine agreement protocol of \cite{Augustine_2013_PODC} and designed a fully-distributed algorithm for {\em Byzantine leader election} that could tolerate up to $O(n^{\frac{1}{2} - \epsilon})$ Byzantine nodes (for any fixed positive constant $\epsilon$)  and up to $\sqrt{n}/\polylog{(n)}$ churn per round and took a $\polylog{(n)}$ number of rounds.

The work of Augustine et al. \cite{Augustine_2022_SPAA} constructs a Distributed Hash Table (DHT) in a Peer-to-Peer (P2P)
network in the presence of a large number (up to $n/\polylog{n}$) of Byzantine nodes. The model used in their paper
assumes a reconfigurable network (where a node can add or drop edges to other nodes whose identifier it knows).
Their paper assumes {\em private channels}, which is significantly weaker than the full information model,
since it assumes that communications between honest nodes are unknown to Byzantine nodes.

The work of Dani et al. \cite{hayes-varsha} studies random walks in a network under the presence of adversarial nodes and devises schemes to detect whether the cover time of random walks can be altered by the behavior of the adversarial nodes.

\section{Byzantine Random Walk Protocol and Theorem}
\label{sec:byzantineRandomWalk}
Generating random walks from all nodes, such that the walks mix yet only few ever visit a Byzantine node (and get corrupted), is a core primitive for fully-distributed Byzantine-tolerant algorithms in sparse (expander) networks. Such a primitive allows to reliably sample the graph, or to set up communication between nodes (e.g., for almost-everywhere reliable information dissemination, see Section \ref{sec:AERID}). We present a protocol for this primitive, the Byzantine random walk protocol, and its properties are captured by the ``Byzantine Random Walk Theorem'' (see Theorem \ref{thm:ByzSamplingFinal}). The protocol and its analysis could be of independent interest. 

\subsection{Definitions and Core}
\label{subsec:defCore}

We consider a network graph $G = (V,E)$ with $|V|=n$ and $|E|=m$, such that at most $|B| = o( n/\log n)$ nodes are Byzantine.
The graph $G$ is assumed to be (1) a regular graph with fixed degree $d$, and (2) an expander graph with constant conductance $\phi_G$ and mixing time $\tau = O(\log n)$. Here, the mixing time is defined as $\tau = \arg \min_t (||A^t \pi - \mathbf{u}||_\infty \le 1/n^3)$, where $A$ is the adjacency matrix of $G$, $\pi$ is any arbitrary probability distribution over $V$, and $\mathbf{u}$ is the stationary distribution over $V$. Note that since $G$ is regular, the stationary distribution $\mathbf{u}$ is uniform, otherwise the stationary probability of a node $u$ would be $deg_G(u)/2m$ instead, where $deg_G(u)$ is the degree of $u$ in $G$.  As mentioned in Section \ref{sec:model}, we assume that nodes
have knowledge of $n$ and the conductance (and hence, the mixing time) of $G$.

Now, if one considers only the honest nodes of $G$ then it is known that a subset of them induces an expander subgraph. More concretely, for a  $d$-regular expander $G = (V,E)$ with $d$ a sufficiently large constant, and at most $|B| = o(n)$ Byzantine nodes, Lemma 3 in \cite{Augustine_2015_FOCS} states that for any chosen constant $c < 1$, there exists a subgraph $C$ in $G \setminus B$ that is of size $n-|B| (1+ \frac{1}{\phi_G (1-c)})= n - O(|B|)$ and that has constant conductance $\phi_C = c\phi_G$.  This expander subgraph of $G$ is called the \emph{core} of $G$, and denoted by $C=(V_C,E_C)$. 
Note that the core $C$ consists of only good nodes, and it need not be a regular graph: i.e., for any $u \in V_C$, $1 \leq deg_C(u) \leq d$, where $deg_C(u)$ is the degree of $u$ restricted to $C$. Moreover, since $C$ is an expander, a random walk restricted to $C$ will have mixing time $\tau_C = b\log n$ for some suitably large constant $b$ (depending on $\phi_C$).
Note that $\tau_C = \arg \min_t (||A_C^t \pi - \mathbf{u}||_\infty \le 1/n^3)$, where $A_C$ is the adjacency matrix of core $C$, $\pi$ is any arbitrary probability distribution over $V_C$, and $\mathbf{u}$ is the stationary distribution over $C$,
defined in Lemma \ref{lem:randwalkcore}. To distinguish $\tau$ and $\tau_C$, we refer to $\tau_C$ as the \emph{core mixing time}.

\begin{lemma}
\label{lem:boundCoreEdges}
    Let $\mu = |B|/|C|$. Then, $|C|/2 \leq (1-O(\mu)) d|C|/2 \leq |E_C| \leq d |C|/2$.
\end{lemma}

\begin{proof}
    The core consists of at least $|C| - O(|B|)$ nodes. Since each node is incident to at most $d$ edges in $G$, $G$ starts with $d |C|/2$ edges and at most $O(d |B|)$ edges are removed to get the core.
\end{proof}

The core subgraph $C$ plays a key role in this section. Indeed, we are particularly interested in the random walks that walk only in this core during our primitive. After all, such walks do not visit any Byzantine nodes, and are likely to mix (rapidly) within the core. 
Additionally, the core also exhibits interesting properties in the subsequent Byzantine agreement protocol. Indeed, this protocol will guarantee (with high probability) that {\em almost all} nodes in the core reach agreement. It is important to note that (honest) nodes themselves do not know whether they belong to $C$.

\subsection{Byzantine Random Walk Protocol}

We next present a distributed protocol to do random walks in a sparse network under the presence of a large number of Byzantine nodes. Our Byzantine Random Walk protocol is presented in Algorithm~\ref{alg:byzantineSamplingSparse}. The protocol addresses the situation where each (honest) node in the network seeks to initiate a number of independent random walks, up to a maximum of $\tot$ random walks per node. (Note that we allow different nodes to initiate different amounts of random walks to allow for a wider range of applications, in particular for those in Subsection \ref{subsec:congestionAERID}.) The protocol operates in phases of $2f = O(\log n)$ rounds each. Each node generates (up to) $\capa = a \log^{3} n$ tokens per phase (for a large enough constant $a \geq 12 \, c \cdot b^2$, where $c$ is the exponent of the whp guarantee and $b$ characterizes the mixing time $\tau_C$ of the core) and these tokens perform independent random walks on $G$. 
Thus, the protocol will require $O(\tot/\capa)$ phases (or more precisely, $\lceil \tot/\capa \rceil$) to complete the process.  

Importantly, each honest node locally regulates the rate at which the tokens flow in and out of it. Specifically,  at most $\capa$ tokens are allowed to enter/exit the node through each of its incident edges per round. We employ a {\em FIFO buffer} at each incident edge to hold tokens that could not be sent in the current round. As a result, a token may be held back at multiple buffers during the phase. Nevertheless, we show in Theorem~\ref{thm:congestion} that all the random walks that {\em only walk on the core $C$} (called ``good'' random walks) will make \emph{at least $f$ random steps} (or in other words, can only be held back during $f$ rounds) whp. 
Then, it follows that if we choose $f$ to be the mixing time of the core $\tau_C = b\log n$ then this will ensure the mixing of those walks in $C$. We then show that most random walks initiated by nodes in $C$ will walk only in $C$ (see Lemma \ref{lem:manyConditionedWalks}). This implies our Byzantine Random Walk Theorem (see Theorem \ref{thm:ByzSamplingFinal}), which says that most random walks initiated in the core $C$ walk only in $C$ and mix rapidly, at which point they reach the stationary distribution over $C$.

\begin{algorithm}[ht]
\caption{Byzantine Random Walk Protocol for node $v$}
\label{alg:byzantineSamplingSparse}
\begin{algorithmic}[1]

\Require 
\Statex $\tot$ \Comment{Maximum (total) number of tokens to be initiated at $v$.}
\Statex $\capa =  a \log^{3} n$ (for large enough constant $a >0$) \Comment{Number of tokens allowed through an edge in one round.}
\Statex $\rwLength = 2f$ \Comment{Length of each phase to ensure good random walks make $f$ steps.}
\Statex $\outbox_u$ for each neighbor $u$ \Comment{FIFO token buffers stored at $v$, one for each neighbor $u$.}

\State $\numPhases = \lceil \tot/\capa \rceil$ 
\For{$\phaseNum \gets 1$ to $\numPhases$}
\State Create $d \cdot \capa$ tokens. Record ID of $v$ as starting vertex.
\ForAll{tokens $\tok$ that were created}
	\State Pick a neighbor $u$ uniformly and independently at random. 
 	\State Push  $\tok$ into $\outbox_u$.
\EndFor
\For{$\stepNum \gets 1$ to $\rwLength$}
\For{each neighbor $u$}
\State Dequeue up to $\capa$ tokens from $\outbox_u$ (which is a FIFO queue).
\State Record $u$ as the next vertex in the walk taken by each of those tokens.
\State Send each dequeued token to $u$.
\EndFor
\State Receive up to $\capa$ tokens sent by each neighbor and store them in a set $M$. 
\Statex \Comment{Any neighbor that sends  $> \capa$ tokens is blacklisted and heretofore ignored.} 
\For{each $\tok$ in $M$}
	\State Pick a neighbor $u$ uniformly and independently at random.
 	\State Enqueue  $\tok$ into $\outbox_u$.
\EndFor
\EndFor
\EndFor
\end{algorithmic}
\end{algorithm}

Let us begin by establishing the round complexity of Algorithm~\ref{alg:byzantineSamplingSparse}.
\begin{lemma}
\label{lem:byzantineRandomWalkRuntime}
    The overall running time of Algorithm~\ref{alg:byzantineSamplingSparse} is $O( f \cdot  \tot/ \capa)$ rounds.
\end{lemma}
\begin{proof}
The algorithm runs for $\tot / \capa$ phases. As each phase takes $2f$ rounds, the algorithm takes $O( f\cdot  \tot/ \capa)$ rounds.
\end{proof}

Next, we focus on the random walks initiated by all nodes in the core $C$ (defined in Subsection \ref{subsec:defCore}) and that walk \emph{only on this subgraph}. We show that if we start \emph{many} independent random walks from each node in $G$ as per Algorithm \ref{alg:byzantineSamplingSparse}, 
then the random walks that \emph{walk only on $C$} will walk at least $f = \tau_C$ steps (the core mixing time).

\begin{theorem} \label{thm:congestion}
Let $C$ be the core of $G$ (consisting of honest nodes only). All random walks initiated at the start of each phase, and that {\em walk only on $C$}, will walk at least $f = \tau_C$ steps (whp) and hence will mix in $C$. 
\end{theorem}

\begin{proof} 
Consider a random walk $w$ that walk only in $C$, or more precisely, that walks on nodes $u_1, u_2, \ldots,u_i, \allowbreak \ldots, u_t$ during the phase, with $u_1$ being the node that initiated $w$ and $u_t$ being the node where it terminated. When the walk enters each $u_i$, there are at most $d\cdot \capa$ random walks that are allowed to enter (from all incident edges) because $u_i$ will discard any excess walks and blacklist any neighbor having sent more than $\capa$ tokens. 

Then, $w$ is placed on $\outbox_{u_{i+1}}$ that was chosen randomly by $u_i$. 
Of course, every other walk (regardless of whether it was initiated by an honest node or a Byzantine node) is also placed randomly in one of the $d$ outboxes. 
Therefore, even assuming the full set of $d \cdot \capa$ walks arrived at $u_i$ along with (and including) $w$, the number of walks placed into $\outbox_{u_{i+1}}$ is a binomial random variable with parameters $d \cdot \capa$ and $1/d$, thus having a mean of $\capa$.
Importantly, when $\capa$ is $a \log^{3} n$, the probability that the number of walks placed into $\outbox_{u_{i+1}}$ will exceed $a \log^{3} n + (a/2b) \log^{2} n = (1+1/(2 b \log n)) \cdot a \log^{3} n$ is at most 
\begin{equation*}
  e^{-(a \log^{3} n)\cdot (1/(4 b^2 \log^2 n))/3} = n^{- a/(12 b^2)}
\end{equation*}
by Chernoff bounds (see Theorem 4.4 in \cite{MitzenmacherUpfalBook}). (Note that since $b \geq 1$, $1/(2 b \log n) \leq 1$ for any $n \geq 2$.) As a result, with probability at most $n^{- a/(12 b^2)}$, 
the excess number of random walks placed into $\outbox_{u_{i+1}}$ (i.e., in addition to the mean $\capa$) is at most $(a/2b) \log^2 n$ per round. With $\rwLength = 2 f = 2 b \log n$, the number of such excess walks placed in $\outbox_{u_{i+1}}$ in the whole phase is at most $2 b \log n \cdot (a/2b) \log^2 n = \capa$ with probability at most $n^{- a/(12 b^2)}$, or in other words, whp for a constant $a$ chosen large enough compared to $b$.

With the excess smaller or equal to $\capa$ (whp), the walk $w$ will buffer at $\outbox_{u_{i+1}}$ for at most one round before moving on to $u_{i+1}$. Thus, even if $w$ were unlucky and took two rounds at each node, in $2f$ steps, it would have taken the requisite $f$ random walk steps to ensure mixing. Thus, we can ensure that all random walks that only walked in $C$ for $2f$ rounds will take at least $f$ random walk steps. 
\end{proof}

Moreover, the random walks that walk only in the core $C$ satisfy the following property: their stationary probability at any node $v \in C$ is within a constant factor of the uniform distribution on $C$.

\begin{lemma}[Random walk conditioned on walking in $C$]
\label{lem:randwalkcore}
Consider, for any core node $v \in C$, any random walk that starts at $v$ and walks \emph{only} through the nodes of $C$ for (core mixing time) $\tau_C$ steps. Then, the probability that the walk is at node $u \in C$ after $\tau_C$ steps (or more) is $p^{\geq \tau_C}(u) = \frac{deg_C(u)}{2|E_C|} \pm \frac{1}{n^3} = \Theta(1/|C|) = \Theta(1/n)$.
\end{lemma}

\begin{proof}
Consider a walk that starts at a node $v \in C$ and walks only on nodes in $C$. Conditioning that the walk only uses edges in $C$, then it holds that at any node in $C$, the walk chose a {\em uniform random outgoing edge among edges in $C$}. Hence, the conditioned random walk is a standard random walk on $C$. Since $C$ is an expander, or more precisely $C$ has constant conductance $\phi_C$, the (conditioned) random walk on $C$ mixes in $\tau_C$ steps and reaches close to the {\em stationary distribution} in $C$ (up to $1/n^3$ error, as we have defined in Subsection \ref{subsec:defCore}). In particular, the stationary probability of node $u \in C$ is $deg_C(u)/(2|E_C|)$. Now, we can apply Theorem \ref{thm:congestion} where $C$ is the honest subset of nodes and $f = \tau_C = b\log n$ is the mixing time of $C$. As a result, the probability that the walk ends up at a node $u \in C$ is proportional to its degree $deg_C(u)$ (where $1 \leq deg_C(u) \leq d$) divided by the number of edges in $C$ (and up to $1/n^3$ error) which is $\Theta(d|C|/2) = \Theta(n)$ (by Lemma \ref{lem:boundCoreEdges}, and as $|B| = o(n)$ and $|C| = n - O(|B|) = \Theta(n)$). Hence, the probability  that the walk ends at $u$  is $\Theta(1/n)$.
\end{proof}

We next show that with high probability, most of the random walks initiated by nodes in $C$ satisfy the conditioning of the above lemma: i.e., with high probability, most walks will walk within $C$.

\begin{lemma}
\label{lem:manyConditionedWalks}
Let $\kappa = (|B| \log n)/|C|$, and let $R(C)$ denote the total number of tokens initiated by the nodes in the core $C$. Recall that each (honest) node initiates a maximum of $\tot$ (good) tokens in Algorithm \ref{alg:byzantineSamplingSparse}. Then, at most $O(\kappa |C| \tot)$ tokens enter or leave $C$, and at least $R(C) - O(\kappa |C| \tot)$ tokens walk only in $C$ (i.e., are good). Moreover, these good tokens walk at least $\tau_C$ steps whp.
\end{lemma}

\begin{proof}
First, we upper bound the number of tokens that enter or leave $C$ in a phasei.e., during the course of the $b\log n$ random walk steps of Algorithm \ref{alg:byzantineSamplingSparse}. To do so, we examine the cut between $C$ and $V-C$. The number of edges crossing this cut is $O(|B|)$. Hence, the number of tokens entering or leaving $C$ during any one round is $O(|B|\log^{3} n)$, since the maximum number of tokens that can go through an edge (from an honest sender or to an honest sender) in any round is $\capa = a \log^{3} n$. Thus, over a phase consisting of $2b\log n$ rounds, the number of tokens entering or leaving $C$ is $O(|B|b\log n \cdot \capa) = O( \kappa |C| \, \capa)$ for $\kappa = (|B| \log n)/|C|$. Over all $\tot/\capa$ phases, the number of tokens entering or leaving $C$ is $O(\kappa |C| \tot)$.

Recall that $R(C)$ denotes the total number of tokens initiated by the nodes in the core $C$. Thus, from the above, $R(C) - O(\kappa |C| \tot)$ tokens walk only in $C$. By the definition of Algorithm \ref{alg:byzantineSamplingSparse} and Theorem \ref{thm:congestion}, all of these good tokens (walking only in $C$) complete $\tau_C = b \log n$ steps in $2\tau_C$ rounds, whp. 
\end{proof}

Theorem \ref{thm:congestion} and Lemmas \ref{lem:byzantineRandomWalkRuntime}, \ref{lem:randwalkcore} and \ref{lem:manyConditionedWalks} together imply the following theorem.

\begin{theorem}[Byzantine Random Walk Theorem] \label{thm:ByzSamplingFinal}
Let $G$ be an expander graph with $n$ nodes, out of which some $o(n/\log n)$-sized subset $B$ of nodes are Byzantine.
Let $C$ be the core of $G$, defined in Subsection \ref{subsec:defCore}, with a mixing time of $\tau_C = b\log n$. Let each (good) node in $G$ (and hence $C$) initiate at most $\tot$ tokens and send them in batches of $\capa = \Theta(\log^{3} n)$ via (independent) random walks for $\rwLength = 2\tau_C$ rounds using Algorithm \ref{alg:byzantineSamplingSparse}. Let $R(C)$ denote the total number of tokens initiated by the nodes in the core $C$. Then the following statements hold whp:
\begin{enumerate}
\item At most $O(\kappa \cdot  \tot \, |C|)$ tokens enter or leave the core.
\item At least $R(C) - O(\kappa \cdot  \tot \, |C|)$ tokens walk only in the core $C$, where $\kappa = (|B| \log n)/|C|$. Moreover, all of these tokens walk at least $\tau_C$ steps whp, and they finish their walks in at most $O(\frac{ \tot}{\capa}\tau_C) = O( \tot/\log^{2}n)$ rounds. 
\item Additionally, the probability that each such token ends at any given node $u \in C$ is $\frac{deg_C(u)}{2|E_C|} \pm \frac{1}{n^3} = \Theta(1/|C|) = \Theta(1/n)$, where $deg_C(u)$ is the degree of node $u$ restricted to the core subgraph $C$ and $E_C$ is the edge set of $C$. 
\end{enumerate}
\end{theorem}

Note that in the Byzantine Random Walk Theorem, $\kappa = (|B| \log n)/|C|$ stands as an upper bound on the fraction of tokens walking out of, or coming into, the core among the (maximum number of) tokens generated by the core. For some applications, it suffices to have $\kappa = o(1)$, or in other words, $|B| = o(n/\log n)$. For others, $\kappa$ should be even smaller. Indeed, in Subsection \ref{subsec:congestionAERID}, we use $\kappa = o(1/\log n)$.

\section{Almost-Everywhere Reliable Information Dissemination}
\label{sec:AERID}
Reliable communication between honest nodes is key to the design of Byzantine-tolerant algorithms, and generally amounts to some variant of reliable broadcast. However, as we have mentioned previously, it is impossible to solve broadcast in sparse networks since Byzantine nodes may isolate a certain number of honest nodes. This motivates us to consider to relax the idea of reliable communication. The relaxed reliable communication is the \emph{almost-everywhere broadcast} primitive, which ensures a node successfully communicates its message to at least $n-o(n)$ honest nodes. A formal definition is given below.

\begin{definition}[Almost-everywhere broadcast]
Let $G=(V,E)$ be a graph on $|V| =n$ nodes out of which up to $o(n)$ nodes can be Byzantine. Let some (honest) node $v \in V$ have a piece of data (say a bit) that it wants to reliably disseminate to all other nodes. A protocol solves {\em Almost-Everywhere Broadcast} if there exists a large enough (receiving) subset $R \subseteq V$ with $|R| = n - o(n)$, such that for any node $u \in R$, $u$ is able to reliably receive the message that originated at $v$ (and associates it with the ID of $v$).
\end{definition}

Rather than a single honest node, it may be the case that many (or all) honest nodes attempt to almost-everywhere broadcast (simultaneously or not). In which case, reliable communication amounts to \emph{Almost-Everywhere to Almost-Everywhere Reliable Information Dissemination (AERID)}: that is, most honest nodes (called broadcasting nodes) can reliably transmit their messages to most honest nodes (called receiving nodes). Interestingly, we do not require that nodes know whether they are part of the broadcasting or receiving subsets. We formally define the AERID primitive below.

\begin{definition}[Almost-Everywhere to Almost-Everywhere Reliable Information Dissemination (AERID)]
Let $G=(V,E)$ be a graph on $|V| =n$ nodes  out of which up to $o(n)$ nodes can be Byzantine. Each (honest) node in $G$ has a piece of data (say a bit) that it wants to reliably disseminate to all other nodes. A protocol solves {\em Almost-Everywhere to Almost-Everywhere Reliable Information Dissemination (AERID)} if there exists a large enough (broadcasting) subset $V' \subseteq V$ with $|V'| = n-o(n)$, and a large enough (receiving) subset $R \subseteq V$ with $|R| = n - o(n)$, such that for any two nodes $u \in R$ and $v \in V'$, $u$ is able to reliably receive the message that originated at $v$ (and associates it with the ID of $v$).
\end{definition}

Note that in the above definition, multiple pieces of data from different nodes can transit through the network simultaneously. In some applications, it is important that a single node's piece of data transits around the network at a time. One such example is our eventual almost-everywhere common coin primitive in Section \ref{sec:randomCoin}. In fact, in Section \ref{sec:randomCoin}, we show that our AERID primitives can serve as a preprocessing step, after which nodes can, one at a time, (attempt to) almost-everywhere broadcast. Crucially, they do so by reusing the communication done during the preprocessing step. And as a result, although not all honest nodes are guaranteed to almost-everywhere broadcast successfully, at least $n-o(n)$ succeed in almost-everywhere broadcasting. 

Subsection \ref{subsec:informationDissemination} presents an AERID protocol with $\tilde{O}(n)$ runtime and tolerating up to $|B| = o(n/ \log n)$ Byzantine nodes. This AERID protocol leads to a costly $\tilde{O}(n)$ runtime for each almost-everywhere broadcast, mainly due to congestion issues, resulting in $\tilde{O}(n)$ runtime for the coin flips in Section \ref{sec:randomCoin} and thus to an $\tilde{O}(n^2)$ runtime for the corresponding almost-everywhere Byzantine agreement (AEBA) protocol (see Section \ref{sec:BASparse}). On the other hand, Subsection \ref{subsec:congestionAERID} gives an AERID protocol with $\tilde{O}(n)$ runtime and tolerating only up to $|B| = o(n/ \log^2 n)$ Byzantine nodes, but improves on the first protocol because it leads to $\tilde{O}(1)$ runtime for the coin flips in Section \ref{sec:randomCoin} and thus to an $\tilde{O}(n)$ runtime for the corresponding AEBA protocol.

\subsection{Fully-Distributed Reliable Information Dissemination}
\label{subsec:informationDissemination}

We start with a relatively simple AERID protocol. It is, in essence, the Byzantine Random Walk protocol (see Algorithm~\ref{alg:byzantineSamplingSparse}) with some slight modifications. In more detail, each node initiates $\Theta(n \log n)$ tokens, which contain (unlike Algorithm~\ref{alg:byzantineSamplingSparse}) its ID and the information to broadcast. Each token executes some $\Theta(\log n)$ length random walks (as indicated by Algorithm~\ref{alg:byzantineSamplingSparse}). Once the algorithm is done, each (honest) node $u$ takes the message broadcast by some node $v$ to be the majority information over all tokens with $v$'s ID received by $u$.

The correctness of this AERID protocol builds upon the following intuition. Say all (honest) nodes initiates $\tot = \Omega(n \log n)$ tokens. and these tokens are distributed according to Algorithm~\ref{alg:byzantineSamplingSparse}. Recall that there exists a core expander subgraph $C$ of $G$ (see Subsection \ref{subsec:defCore}) containing only honest nodes, and that walks (or tokens) in Algorithm~\ref{alg:byzantineSamplingSparse} are said to be good if they start and walk only in $C$. Then the Byzantine Random Walk Theorem (see Theorem \ref{thm:ByzSamplingFinal}) ensures that most of the tokens (executing random walks of length $\Theta(\tau_C) = \Theta(\log n)$) initiated by good nodes in $C$ remain and mix in the core $C$. As a result of never visiting any Byzantine node, their content (i.e., ID and information to broadcast) is never corrupted.
Additionally, the mixing and the choice of $\tot$ ensures that these good tokens (initiated by nodes in $C$ and walking only in $C$) will evenly spread (and the uncorrupted information they contain) out through $C$ in good enough quantities. On the other hand, because there are far fewer bad tokens (i.e., tokens that exit or enter $C$) than good tokens, Byzantine nodes cannot spread out (bad) information in the same quantities to as many nodes in $C$. Next, we give a formal statement (Theorem \ref{thm:AEtoAEBroadcast}) capturing the above intuitions.

\begin{theorem}
\label{thm:AEtoAEBroadcast}
    With high probability there exists two large enough subsets, a broadcasting subset $C_b \subseteq C$ with $|C_b|  = |C| - o(|C|)$, and a receiving subset $C_r \subseteq C$ with $|C_r|  = |C| - o(|C|)$, such that for any $v \in C_b$ and $u \in C_r$, $u$ receives $\Theta(\log n)$ good tokens initiated by $v$ and $o(\log n)$ bad tokens originating supposedly from $v$.
\end{theorem}

\begin{proof} We build up to the theorem's statement by showing the following two claims.
    \begin{claim}
    \label{claim:goodTokensSpreading}
        With high probability there exists a large enough (broadcasting) subset $C' \subseteq C$, with $|C'| = |C| - o(|C|)$, such that for any node $v \in C'$, every node $u \in C$ receives $\Theta(\log n)$ good tokens initiated by $v$.
    \end{claim}

    \begin{proof}
        By a simple counting argument and Item (2) of Theorem \ref{thm:ByzSamplingFinal} (with $\kappa = o(1)$), there exist a set $C'$, with $|C'| = |C| - o(|C|)$, such that for each node $v \in C'$, at least $n \log n - o(n \log n)$ of the tokens originating at $v$ walk only in $C$.
        Now, by Item (3) of Theorem \ref{thm:ByzSamplingFinal}, each such walk will end at any given node $u \in C$ with probability $\Theta(1/|C|) = \Theta(1/n)$. Hence, any given node $u$ receives $\Theta((n \log n -o(n \log n)) / n) = \Theta(\log n)$ tokens from each node $v \in C'$ in expectation. Since the random walks are independent, standard Chernoff bounds (see Theorem 4.4 in \cite{MitzenmacherUpfalBook}) imply that with high probability, for any node $v \in C'$ and any node $u \in C$, $u$ receives $\Theta(\log n)$ tokens initiated by $v$.
    \end{proof}

    \begin{claim}
    \label{claim:badTokens}
        With high probability there exists two large enough subsets, a broadcasting subset $C_b \subseteq C$ with $|C_b|  = |C| - o(|C|)$, and a receiving subset $C_r \subseteq C$ with $|C_r|  = |C| - o(|C|)$, such that for any $v \in C_b$ and $u \in C_r$, $u$ receives $o(\log n)$ bad tokens originating supposedly from $v$. 
    \end{claim}

\begin{proof}
    By Item (1) of Theorem \ref{thm:ByzSamplingFinal} at most $O(\kappa (n \log n) |C|)$ tokens enter (or leave) $C$ during the algorithm, where $\kappa = (|B| \log n)/|C|$. Next, recall that $|B|=o(n/\log n)$. Thus, $o((n \log n) |C|)$ bad tokens enter $C$ during the algorithm.    
    
    Finally, a counting argument suffices to obtain the claim statement. We now describe this argument in detail.
    Let us fix a (broadcasting) subset $C_b \subseteq C$ with $|C_b|  = |C| - o(|C|)$. We will show that there exists a (receiving) subset $C_r \subseteq C$ with $|C_r|  = |C| - o(|C|)$, such that for any $v \in C_b$ and $u \in C_r$, $u$ receives $o(\log n)$ bad tokens originating (supposedly) from $v$. Consider a subset $\hat{C}_r$ of $C$ where each node $u \in \hat{C}_r$ receives at least $\Omega(\log n)$ bad tokens (supposedly) from some node $v \in C_b$. We show by contradiction that $\hat{C}_r = o(|C|)$ and hence the subset $C_r$ can be taken to be $C-\hat{C}_r$. Suppose not, let $\hat{C}_r = \Theta(|C|)$. Then, by the above assumption, in total at least $\sum_{u \in \hat{C}_r, v \in C_b} \Omega(\log n) = \Omega((\log n) ((|C|-o(|C|))|C|)) = \Omega((n \log n) |C|)$ bad tokens are received by nodes in $\hat{C}_r$. (Recall that $|C| \geq n - O(|B|) = n - o(n)$.) However, only $o((n \log n) |C|)$ bad tokens enter $C$, leading to a contradiction.
    \end{proof}

    By Claim \ref{claim:badTokens}, with high probability there exists two large enough subsets, a broadcasting subset $C_b \subseteq C$ with $|C_b  = |C| - o(|C|)$, and a receiving subset $C_r \subseteq C$ with $|C_r  = |C| - o(|C|)$, such that for any $v \in C_b$ and $u \in C_r$, $u$ receives $o(\log n)$ bad tokens originating (supposedly) from $v$. 
    Moreover, by Claim \ref{claim:goodTokensSpreading}, there exists with high probability a large enough (broadcasting) subset $C'' \subseteq C$, with $|C''| = |C| - o(|C|)$, such that for any node $v \in C''$, every node $u \in C$ receives $\Theta(\log n)$ good tokens originating in $v$.
    Then, there exists a large enough (broadcasting) subset $C^* = C_b \cap C''$, with $|C^*| = |C| - o(|C|)$, and a large enough (receiving) subset $C_r$, such that for any $v \in C^*$ and $u \in C_r$, $u$ receives $\Theta(\log n)$ good tokens and $o(\log n)$ bad tokens originating (supposedly) from $v$.
\end{proof}

The correctness and runtime of this simple AERID protocol is captured by the below theorem; the correctness follows from Theorem \ref{thm:AEtoAEBroadcast} and the $\tilde{\Theta}(n)$ runtime from Theorem \ref{thm:ByzSamplingFinal}.

\begin{theorem}
\label{thm:informationDissemination}
    AERID can be solved with high probability in an $n$-node $d$-regular expander network with up to $|B| = o(n/\log n)$ Byzantine nodes in $\tilde{O}(n)$ rounds.
\end{theorem}

Once again, we point out that nodes do not need to know if they are part of the broadcasting or receiving subsets to solve the AERID primitive. And in fact, in the presented AERID protocol, all nodes ''attempt'' to be in both subsets, and those who fail are not aware of it.

A major disadvantage of the presented AERID protocol lies in its congestion issues, even when we consider only the $\Theta(n \log n)$ good tokens that contain any given node's ID. Indeed, all honest nodes generate $\Theta(n \log n)$ random walks (or communication paths) but are only incident to $O(1)$ edges. Therefore, the congestion of the good tokens on these edges is $\Theta(n \log n)$. This leads to a significant  $\Theta(n \log n)$ slowdown caused by congestion when this protocol is used in our common coin primitive, in Section \ref{sec:randomCoin}.

\subsection{Low-Congestion Reliable Information Dissemination}
\label{subsec:congestionAERID}

We now present a more sophisticated AERID protocol, which tolerates less Byzantine nodes but takes care of the congestion issues encountered in the previous AERID protocol.
Let us give the intuition underlying this protocol. In it, each honest node will be the ``origin'' of $\Theta(n \log n)$ tokens, but only $\tilde{\Theta}(1)$ are initially generated by the node. These initially generated tokens execute random walks and reach some randomly chosen (honest) nodes (for most tokens), at which point they are duplicated, say twice. All tokens then once again execute random walks, and get duplicated again, and so on. After roughly $\log n$ stages, the tokens ``originating'' at any one node will have grown to $\Theta(n \log n)$ and spread around the network, similarly to the previous protocol. However, unlike the case where all $\Theta(n \log n)$ ``originating'' from a node are generated by the same node, here these tokens are generated by all of the nodes, in an almost uniform fashion, and this ensures that no edge will ensure the transit of too many tokens with the same ``origin'' in any one stage.

More formally, the protocol runs in $O(\log n)$ stages. Each stage $i \geq 1$ consists of running the random walk protocol, such that each node initiates at most $\tot(i) = \lambda (2 (1+\delta))^{i}$ for any arbitrarily chosen small $0 \leq \varepsilon \leq 1$ and some well-chosen constants $\lambda = \Theta(\log^{3+2\varepsilon} n)$ and $\delta = \Theta(1/\log^{1+\varepsilon} n)$. Note that this means that stages will run for increasingly longer times, and importantly, the constants are set such that for the last stage $i^*$, $\tot(i^*) = \Theta(n \log n)$ (and thus the stage will take $\tilde{O}(1)$ rounds). 
Initially, each node $v$ holds $\tot(0)= \lambda$ tokens, containing the node's ID, the message $m_v$ that node $v$ is seeking to disseminate, and some auxiliary stage-tracking and step-tracking counters set to 0. Note that no honest node will ever modify the ID or message contained in a token. We call that ID the token's \emph{source ID}.
Next, consider a given stage $i \geq 1$. In it, each node $v$ first considers the tokens it held at the end of the previous stage (or initially if $i = 1$). If $v$ holds more than $(1+\delta) \tot(i-1)$ tokens (not necessarily with the same source ID), then $v$ keeps only $(1+\delta) \tot(i-1)$ arbitrarily chosen ones and deletes the rest. 
The stage counter of each kept token is changed to $i$, and its step counter is reset to $0$, after which the token is duplicated twice and the original token is then discarded (this simply makes some of the definitions below easier). In other words, the tokens held by nodes are now doubled. Finally, all of these duplicate tokens are sent out via the Byzantine Random Walk protocol; whenever a token moves, its step counter is updated accordingly. 

For the analysis, we start by proving the correctness and runtime of the AERID protocol. First, we extend the definition of a good token, given in the previous subsection, in a natural way for this analysis. Indeed, we say here that a token is \emph{good} if it walks only in $C$ and its parent token (from which it was duplicated) was also good. Moreover, a good token is said to be \emph{low-congestion} in stage $i \geq 1$ if its parent token (from which it was duplicated) was low-congestion and if any edge it traverses during the $j$th step of its random walk of stage $i$ (for any step $j \geq 1$, $j \leq 2 \tau_C = O(\log n)$) is traversed by at most $O(\lambda / \delta) = O(\log^{4+3\varepsilon} n)$ other tokens (whether good or bad) with the same source ID during that same step within that same stage.

Next, note that the duplication of tokens done at the start of every stage in this protocol (unlike the previous one) increases the number of good tokens in every stage. We start by upper bounding the growth of the number of good tokens with a given source ID over the multiple stages. Note that by definition of a good token, the source ID must be from $C$. 

\begin{lemma}
\label{lem:tokenUpperBound}
    For any stage $i \geq 1$, there are $O(\lambda 2^{i})$ good tokens with a given source ID both when the stage starts (after duplication) and when the stage ends.
\end{lemma}

\begin{proof}
    A simple induction on $i$ suffices. Indeed, after the first stage's duplication there are $2 \lambda$ tokens for any given source ID from $C$. After which, each stage leads to tokens being discarded or becoming bad, and the tokens that remain good are duplicated twice. 
\end{proof}

We can also upper bound the number of good tokens with a given source ID that start or end some stage at any one given node by $\tilde{O}(1)$. Note that the same upper bound holds regardless of the stage counter.

\begin{lemma}
\label{lem:goodTokenDispersion}
    For any stage $i \geq 1$ and any node $v \in C$, with high probability, any node $u \in C$ starts the stage with $O(\lambda)$ good tokens with the source ID of $v$ and ends the stage with $O(\log n)$ tokens with the source ID of $v$.
\end{lemma}

\begin{proof}
    We first show the later half of the statement. Note that by Lemma \ref{lem:tokenUpperBound}, for any stage $i \geq 1$, there are $O(\lambda 2^{i})$ good tokens with a given source ID both when the stage starts (after duplication) and when the stage ends. Those tokens that start and remain good throughout stage $i$ by definition walk only in $C$. By Theorem \ref{thm:ByzSamplingFinal}, each such random walk mixes in $C$, or more precisely, ends up at some node $u \in C$ with probability $\Theta(1/n)$. Now, let $X_{u,v}$ denote the number of tokens received by some node $u \in C$ and that hold the ID of some $v \in C$ as a source ID. From how the total number of stages is defined, $O(\lambda 2^{i}) = O(n \log n)$ for any stage $i \geq 1$. Then, the expectation $\E[X_{u,v}]$ can be upper bounded by some $k = \Theta(\log n)$, and since these tokens execute independent random walks, we can use Chernoff bounds modified to use an upper bound on the expectation (see \cite{DubhashiPanconesi}): 
    \begin{equation*}
        \Pr[X_{u,v} \geq (1+1/2) k] \leq e^{- k \cdot (1/2)^2 /3}
    \end{equation*}
    This implies that any node $u \in C$ receives $O(\log n)$ tokens with a given source ID whp when any stage $i \geq 1$ ends.

    As for the first half of the statement, it is now straightforwardly obtained. For the first stage, note that every node $u \in C$ starts the first stage (after duplication) with at most $O(\lambda)$ good tokens with a given source ID (from $C$). (In fact, only the node with that ID starts with that many tokens, all others start with none.) As for the later stages $i > 1$, the second half of the lemma statement implies that every node $u \in C$ starts any stage $i > 1$ (after duplication) with at most $2 k = \Theta(\log n)$ good tokens with a given source ID whp.
\end{proof}

As a result, for any given stage, few good tokens traverse the same edge during the $j$th step of the random walk they execute in that stage. In other words, if all tokens in the network were good, then these tokens would all also be low-congestion.

\begin{lemma}
\label{lem:goodTokenEdgeCongestion}
    For any stage $i \geq 1$, any step $j$ and for any node $v \in C$, with high probability any given edge in $G$ is traversed at most $O(\lambda)$ times by good tokens with the source ID of $v$ (or more precisely, good at the time of traversal).
\end{lemma}

\begin{proof}
    By Lemma \ref{lem:goodTokenDispersion}, every node $u \in C$ starts the stage (after duplication) with at most $k' = \Theta(\lambda)$ good tokens with a given source ID (from $C$). Let $X_{u,e}^j$ (respectively, $Y_{u,e}^j$) be the number of tokens node $u$ sends (resp., receives) on any edge $e$ among its $d$ incident edges (including those leading outside the core) for the token's $j$th step, that are still good (i.e, have walked only in the core up to, and including, step $j$ of stage $i$) and contain the source ID of some node $v \in C$. Then, one can see that $\E[X_{u,e}^j]$ and $\E[Y_{u,e}^j]$ are upper bounded by $k'/d = \Theta(\lambda)$ for any node $u \in C$, incident edge $e$ and step $j$. The number of good tokens (i.e., good during traversal) traversing edge $e$ in their $j$th step (with the source ID of $v$) is $X_{u,e}^j + Y_{u,e}^j$ and its expectation is upper bounded by $2 k' / d = \Theta(\lambda) = \Omega(\log^3 n)$. Using Chernoff bounds modified to use an upper bound on the expectation (see \cite{DubhashiPanconesi}), we can show that:
    \begin{equation*}
        \Pr[X_{u,e}^j + Y_{u,e}^j \geq (1+1/2) (2 k' / d)] \leq e^{- (2 k' / d) \cdot (1/2)^2 /3}
    \end{equation*}    
    Therefore, we get that with high probability, at most $O(\lambda)$ good tokens (with the source ID of a given node in $C$) traverse any given edge in their $j$th step during the $i$th stage's random walk (or at least tokens that are still good when traversing that edge for the $j$ step).
\end{proof}

Now, if all tokens remained good (i.e., did not walk out of the core), and no token was discarded due to overcapacity (i.e., when there are more than $(1+\delta) \tot(i-1)$ tokens at a given node at the start of stage $i$), the upper bound of Lemma \ref{lem:tokenUpperBound} would be tight for all source IDs. And once again, if all tokens remained good, then all these good tokens would be low-congestion. This may not hold but we can show a slightly weaker statement (see Lemma \ref{lem:tokenTightBound}). To do so, we must first show that throughout the protocol, there always are sufficiently many good tokens whp (see Lemma \ref{lem:tokenUniformDispersal} below). 

\begin{lemma}
\label{lem:tokenUniformDispersal}
    Recall that $\delta = \Theta(1/\log^{1+\varepsilon} n)$, and let $\kappa = (|B| \log n)/|C|$. Let $\eta = \kappa + \delta$.
    Then, for any stage $i \geq 1$, the nodes of the core start the stage (after duplication) with at least $(1-O(\eta) \cdot i) |C| \tot(i)$ good and low-congestion tokens whp, and end the stage with at least $(1-O(\eta) \cdot (i+1)) |C| \tot(i)$ good and low-congestion tokens whp (where the $O$ notation hides constants that are independent of $i$).
\end{lemma}

\begin{proof}
    We prove the first half of the statement by induction on $i \geq 1$. (The second half can be obtained in a similar fashion.) The base case for stage $i=1$ is straightforward: initially, each node in the core holds $\tot(0)$ good and low-congestion tokens, by definition. Each of these tokens is duplicated twice. The result is that $2 \; \tot(0) = \tot(1) /(1+\delta) = (1-O(\delta)) \tot(1)$ tokens are held by the core nodes after duplication in the first stage, and these tokens remained good and low-congestion.
    
    As for the induction step, suppose $i > 1$ and consider stage $i-1$ after the duplication. Then, by the induction hypothesis for $i-1$, it holds that at least $(1-O(\kappa + \delta) (i-1)) |C| \tot(i-1)$ good and low-congestion tokens are held by the nodes in the core. We aim to upper bound, among these tokens, how many (a) exit the core (and thus become bad), (b) are discarded because of overcapacity at a given node, or (c) are no longer low-congestion.
    Part (a) is easy to bound, since by Item (1) of Theorem \ref{thm:ByzSamplingFinal}, at most $O(\kappa |C| \tot(i-1))$ of all tokens originating in $C$ exit $C$. 
    
    On the other hand, part (b) is slightly more involved. Indeed, a node may hold more than $(1+\delta) \tot(i-1)$ tokens because it holds many bad tokens (i.e., bad tokens entering $C$ during stage $i-1$) but also due to purely probabilistic considerations (i.e., too many random walks end at the same node). We first upper bound the good tokens being discarded in favor of bad tokens. For this, it suffices to notice that any one bad token leads to at most one good token being discarded. By Item (1) of Theorem \ref{thm:ByzSamplingFinal}, at most $O(\kappa |C| \tot(i-1))$ bad tokens enter $C$ during stage $i-1$. Then, these bad tokens force at most $O(\kappa |C| \tot(i-1))$ good tokens to be discarded in the core. 
    
    To prove (b), it remains to upper bound the number of good tokens being discarded in favor of good tokens. Note that there are at most $|C| \tot(i-1)$ good tokens in stage $i-1$, and those that remain good throughout stage $i-1$ execute independent random walks within $C$ only, by definition. By Theorem \ref{thm:ByzSamplingFinal}, each such random walk mixes in $C$, or more precisely, ends up at some node $u \in C$ with at most probability $deg_C(u)/(2 |E_C|) + 1/n^3$. By Lemma \ref{lem:boundCoreEdges}, this probability is at most $1/((1-O(\mu)) |C|) + 1/n^3 = (1+O(\mu))/|C| + 1/n^3$ for $\mu = |B| / |C| = o(1/\log n)$. 
    Hence, the expected number of good tokens received by any node in $C$ is at most $(1+O(\mu) + 1/n^2) \tot(i-1) = (1+O(\mu)) \tot(i-1)$. 
    Moreover, the random walks are independent so we can use Chernoff bounds (modified to use an upper bound on the expectation, see \cite{DubhashiPanconesi}). We show that the number of tokens $X$ received by any node in $C$ satisfies
    \begin{equation*}
        \Pr[X \geq (1+\delta) (1+O(\mu)) \tot(i-1)] \leq e^{- (1+O(\mu)) \tot(i-1) \cdot \delta^2 /3}
    \end{equation*}
    
    Given that $\tot(i-1) \geq \lambda = \Theta(\log^{3+2\varepsilon} n)$ for any $i > 1$ and $\delta^2 = \Theta(1 / \log^{2+2\varepsilon} n)$, it holds that $(1+O(\mu)) \tot(i-1) \cdot \delta^2 /3 = \Omega(\log n)$. To sum up, any node in $C$ receives at most $(1+O(\mu)+O(\delta)) \tot(i-1)$ good tokens whp. As a result, $(O(\mu+\delta)) \tot(i-1) = (o(\kappa)+O(\delta)) \tot(i-1)$ good tokens are discarded (due to good tokens) whp.  

    Finally, we need to upper bound how many of the good and low-congestion tokens do not remain low-congestion throughout the stage. By Lemma \ref{lem:goodTokenEdgeCongestion}, for any step $j$ and for any node $v \in C$, any given edge in $G$ is traversed at most $O(\lambda)$ times by good tokens with the source ID of $v$ (or more precisely, good at the time of traversal) during stage $i-1$.
    Consider, for some edge, all $O(\lambda)$ good tokens that take that same edge in step $j$. Then, these good tokens may fail to be low-congestion in step $j$ only if at least $\Omega(\lambda/\delta)$ tokens (and thus bad tokens) take that edge in step $j$. Now, Item (1) of Theorem \ref{thm:ByzSamplingFinal} implies that at most $O(\kappa |C| \tot(i-1)) = o(|C| \tot(i-1))$ bad tokens walk within $C$ during stage $i-1$. Altogether, they can over-congest at most $o( \delta |C| \tot(i-1) )$ good tokens for any step $j$ (or in other words, make it so that the good token is not low-congestion in step $j$), and thus $o(\delta |C| \tot(i-1))$ good tokens over all steps $j$ of stage $i-1$. 
    
    In total, at least $(1-O(\eta) (i-1)) |C| \tot(i-1) - O(\eta) |C| \tot(i-1) \geq (1- O(\eta) \cdot i) |C| \tot(i-1)$ good tokens are held by nodes in the core at the end of stage $i-1$. (Note here that the constants hidden by the big O notation are crucially independent of $i$.) 
    Finally, during the duplication in stage $i$, each (good) token is duplicated twice and thus afterwards at least $(1- O(\eta) \cdot i) |C| 2 \; \tot(i-1) = (1- O(\eta) \cdot i) |C| \tot(i) / (1+\delta) = (1- O(\eta) \cdot i) |C| \tot(i)$ tokens are good whp, or in other words, the induction step follows.     
\end{proof}

Choosing $\kappa$ low enough (and since $\delta = o(1/\log n)$), we get the following corollary from the above lemma. In essence, if there are $o(n/\log^2 n)$ Byzantine nodes, most of the tokens (starting and ending) in the core are good and low-congestion in all stages. 

\begin{corollary}
\label{cor:tokenUniformDispersal}
    Let $\kappa = (|B| \log n)/|C| = o(1/\log n)$. Then, for any stage $i \geq 1$, the nodes of the core start the stage (after duplication) with at least $(1-o(1)) |C| \tot(i)$ good and low-congestion tokens whp, and end the stage also with at least $(1-o(1)) |C| \tot(i)$ good and low-congestion tokens whp (where the $o$ notation hides constants that are independent of $i$).
\end{corollary}

Now that we have a lower bound on the number of good and low-congestion tokens over all possible source IDs in every stage, we can prove that there exists many source IDs for which there are many good and low-congestion tokens in every stage.

\begin{lemma}
\label{lem:tokenTightBound}
    Let $\kappa = (|B| \log n)/|C| = o(1/\log n)$. Then, with high probability, there exists a large enough subset $C' \subseteq C$, with $|C'| = |C| - o(|C|)$, such that for any stage $i \geq 1$, there are $\Theta(\lambda 2^{i})$ good and low-congestion tokens with a given source ID (from a node in $C'$) both when the stage starts (after duplication) and when the stage ends.
\end{lemma}

\begin{proof}
    By Lemma \ref{lem:tokenUpperBound}, we know that the number of good tokens with a given source ID (by definition from a node in $C$) is $O(\lambda 2^{i})$ for any stage $i \geq 1$, and this number can only decrease throughout stage $i$. Of these, there can only be less that are also low-congestion.

    Next, note that by Corollary \ref{cor:tokenUniformDispersal}, for any stage $i \geq 1$, the nodes of the core start the stage (after duplication) with at least $(1-o(1)) |C| \tot(i)$ good and low-congestion tokens whp, and end the stage also with at least $(1-o(1)) |C| \tot(i)$ good and low-congestion tokens tokens whp. As $\tot(i) = \lambda 2^{i} (1+\delta)^{i}$ and $(1+\delta)^{i} = e^{i \ln(1+\delta)} = e^{i (\delta + o(\delta))} = 1+o(1)$ for $\delta=o(1/\log n)$ and $i = O(\log n)$, we get that $(1-o(1)) |C| \tot(i) = (1-o(1)) |C| \lambda 2^{i}$.
    However, this contradicts the existence of some $C_{bad} \subseteq C$ with $|C_{bad}| = \Omega(|C|)$ such that the number of good and low-congestion tokens with a source ID in $C_{bad}$ is $o(\lambda 2^{i})$. Hence, for any stage $i$, there exists a large enough subset $C'(i) \subseteq C$, with $|C'(i)| = |C| - o(|C|)$, such that there are $\Theta(\lambda 2^{i})$ good and low-congestion tokens with a given source ID (from a node in $C'(i)$) both when the stage starts (after duplication) and when the stage ends.
    
    Finally, first note that if the number of good and low-congestion tokens with a given source ID (by definition from a node in $C$) is $o(\lambda 2^{i})$ for some stage $i \geq 1$, then there are (at most) $o(\lambda 2^{i'})$ good and low-congestion tokens with the same source ID for any later stages $i'> i'$ since the number of good (and low-congestion) tokens with the same source ID increases by at most $\lambda$ per stage. As a result, we can consider the claim shown at the end of the above paragraph, and take the large enough subset $C'$ satisfying the claim for the last stage. It follows that for any stage $i \geq 1$, there are $\Theta(\lambda 2^{i})$ good and low-congestion tokens with a given source ID (from a node in $C'$) both when the stage starts (after duplication) and when the stage ends.
\end{proof}

Having set up the previous lemmas on the amount of good and low-congestion tokens moving around the core in every stage, both overall as well as from any given source ID, we can now prove the correctness of the presented AERID protocol (see Theorem \ref{thm:correctnessCongestionAERID} below). In particular, this proves correctness even when we consider only the good and low-congestion tokens, rather than all good tokens.

\begin{theorem}
\label{thm:correctnessCongestionAERID}
    Let $\kappa = (|B| \log n)/|C| = o(1/\log n)$. Then with high probability there exists two large enough subsets, a broadcasting subset $C_b \subseteq C$ with $|C_b| = |C| - o(|C|)$, and a receiving subset $C_r \subseteq C$ with $|C_r| = |C| - o(|C|)$, such that in the last stage, node $v \in C_b$ and $u \in C_r$, $u$ receives $\Theta(\log n)$ good and low-congestion tokens with the source ID of $v$, and $o(\log n)$ also with that source ID.
\end{theorem}

\begin{proof}
    This proof follows along the lines of that of Theorem \ref{thm:AEtoAEBroadcast}. We start by showing the following claims. 

    \begin{claim}
    \label{claim:broadcastingEfficient}
        Let $\kappa = (|B| \log n)/|C| = o(1/\log n)$. Then with high probability there exists a large enough subset $C' \subseteq C$, with $|C'| = |C| - o(|C|)$, such that every node $u \in C$ ends the last stage with $\Theta(\log n)$ good and low-congestion tokens with the source ID of $v$.
    \end{claim}
    
    \begin{proof}
    By Lemma \ref{lem:tokenTightBound}, there exists a large enough subset $C' \subseteq C$, with $|C'| = |C| - o(|C|)$, such that for the last stage, there are $\Theta(n \log n)$ good and low-congestion tokens with a given source ID (from some node $v \in C'$) both when the stage starts (after duplication) and when it ends. 
    Those tokens that start and remain good throughout the last stage by definition walk only in $C$. By Item (3) of Theorem \ref{thm:ByzSamplingFinal}, each such random walk mixes in $C$, or more precisely, ends up at some node $u \in C$ with probability $\Theta(1/n)$. Now, let $X_{u,v}$ denote the number of (good and low-congestion) tokens received by some node $u \in C$ that hold the ID of $v \in C'$ as a source ID. Then, the expectation $E[X_{u,v}] = \Theta(\log n)$ and since these tokens execute independent random walks, we can use standard Chernoff bounds (see Theorem 4.4 in \cite{MitzenmacherUpfalBook}) to prove that for any node $v \in C'$, every node $u \in C$ receives $\Theta(\log n)$ good and low-congestion tokens with the source ID of $v$ when the last stage ends.
    \end{proof}
    
    \begin{claim}
    \label{claim:notTooManyBadReceptionsEfficient}
        Let $\kappa = (|B| \log n)/|C| = o(1/\log n)$. Then with high probability there exists two large enough subsets, a broadcasting subset $C_b \subseteq C$ with $|C_b| = |C| - o(|C|)$, and a receiving subset $C_r \subseteq C$ with $|C_r| = |C| - o(|C|)$, such that in the last stage, node $v \in C_b$ and $u \in C_r$, $u$ receives $o(\log n)$ bad tokens with the source ID of $v$. 
    \end{claim}
    
    \begin{proof}
    In the last stage $i^*$, $\tot(i^*) = \Theta(n \log n)$ and thus nodes in the core initiate (at most) $|C| \tot(i^*) = \Theta(|C| (n \log n))$. Of these, by Corollary \ref{cor:tokenUniformDispersal}, at most $o(|C| \tot(i^*))$ are bad tokens. Additionally, by Item (1) of Theorem \ref{thm:ByzSamplingFinal}, at most $O(\kappa |C| \tot(i^*)) = o(|C| \tot(i^*))$ tokens enter $C$ during stage $i^*$. Adding the two together, at most $o((n \log n) |C|)$ bad tokens end stage $i^*$ at a node in the core. 

    Finally, a counting argument (similarly to Claim \ref{claim:badTokens}) suffices to obtain the lemma statement. More concretely, let us fix a (broadcasting) subset $C_b \subseteq C$ with $|C_b|  = |C| - o(|C|)$. We will show that there exists a (receiving) subset $C_r \subseteq C$ with $|C_r|  = |C| - o(|C|)$, such that for any $v \in C_b$ and $u \in C_r$, $u$ receives $o(\log n)$ bad tokens originating (supposedly) from $v$. Consider a subset $\hat{C}_r$ of $C$ where each node $u \in \hat{C}_r$ receives at least $\Omega(\log n)$ bad tokens with the source ID of some node $v \in C_b$. We show by contradiction that $\hat{C}_r = o(|C|)$ and hence the subset $C_r$ can be taken to be $C-\hat{C}_r$. Suppose not, let $\hat{C}_r = \Theta(|C|)$. Then, by the above assumption, in total at least $\sum_{u \in \hat{C}_r, v \in C_b} \Omega(\log n) = \Omega((\log n) ((|C|-o(|C|))|C|)) = \Omega((n \log n) |C|)$ bad tokens are received by nodes in $\hat{C}_r$. (Recall that $|C| \geq n - O(|B|) = n - o(n)$.) However, only $o((n \log n) |C|)$ bad tokens enter $C$, leading to a contradiction.
    \end{proof}

    By Claim \ref{claim:notTooManyBadReceptionsEfficient}, with high probability there exists two large enough subsets, a broadcasting subset $C_b \subseteq C$ with $|C_b  = |C| - o(|C|)$, and a receiving subset $C_r \subseteq C$ with $|C_r  = |C| - o(|C|)$, such that in the last stage for any $v \in C_b$ and $u \in C_r$, $u$ receives $o(\log n)$ bad tokens with the source ID of $v$. 
    Moreover, by Claim \ref{claim:broadcastingEfficient}, there exists with high probability a large enough (broadcasting) subset $C'' \subseteq C$, with $|C''| = |C| - o(|C|)$, such that for any node $v \in C''$, every node $u \in C$ receives $\Theta(\log n)$ good and low-congestion tokens with the source ID of $v$.
    Then, there exists a large enough (broadcasting) subset $C^* = C_b \cap C''$, with $|C^*| = |C| - o(|C|)$, and a large enough (receiving) subset $C_r$, such that for any $v \in C^*$ and $u \in C_r$, $u$ receives $\Theta(\log n)$ good and low-congestion tokens with the source ID of $v$, and $o(\log n)$ bad tokens that also have the source ID of $v$.
\end{proof}

Now, we can show that the presented protocol solves AERID. Moreover, it does so even if we discard all good tokens that are not low-congestion (and their contained information). This additional property is crucial for our faster eventual almost-everywhere common coin primitive.

\begin{theorem}
\label{thm:informationDisseminationCongestionFree}
    AERID can be solved with high probability in an $n$-node $d$-regular expander network with up to $|B| = o(n/\log^2 n)$ Byzantine nodes in $\tilde{O}(n)$ rounds. Moreover this holds even if  all good tokens (and thus the information they carry) except the low-congestion good tokens are discarded.
\end{theorem}

\begin{proof}
    The correctness (and the later half of the statement) follows from Theorem \ref{thm:correctnessCongestionAERID}. On the other hand, the runtime on the other hand follows from applying Item (2) of Theorem \ref{thm:ByzSamplingFinal} to each stage. More concretely, each node runs the Byzantine Random Walk protocol in every phase $i \geq 1$, initiating up to $\tot(i) = \Theta(\lambda 2^{i} n)$ tokens and thus taking $\tilde{O}(\tot(i))$ rounds. Since $\tot(i) = \tilde{O}(n)$ for all phases, the runtime is $\tilde{O}(n)$.
\end{proof}

\section{Eventual Almost-Everywhere Common Coin}
\label{sec:randomCoin}

Common coin primitives lie at the heart of many Byzantine agreement primitives. They allow all nodes to generate a \emph{common} random bit. In this section, we design an eventual almost-everywhere common coin (EAECC) primitive. By almost-everywhere (common), we mean to say that almost all honest nodes (i.e., at least $n - o(n)$ honest nodes) agree on the random bit. Whereas by eventual, we mean to say that within $n$ coin flips (i.e., calls to this primitive), at least one flip will be random (i.e., it is 0 or 1 each with equal probability).

We leverage the two AERID primitives from Section \ref{sec:AERID} to design two eventual, almost-everywhere random coin primitives The first coin primitive, described in Subsection \ref{subsec:coinWithBetterTolerance}, builds upon the first AERID primitive from Subsection \ref{subsec:informationDissemination}. It tolerates up to $o(n/ \log n)$ Byzantine nodes, but has a slow runtime. Whereas the second coin primitive, described in Subsection \ref{subsec:fasterCoin} and building upon the second AERID primitive from Subsection \ref{subsec:congestionAERID}, tolerates only up to $o(n/ \log^2 n)$ Byzantine nodes but is significantly faster (by a linear in $n$ factor in fact).

\subsection{Eventual Almost-Everywhere Common Coin with \texorpdfstring{$o(n/\log n)$}{o(n/ log n)} Tolerance}
\label{subsec:coinWithBetterTolerance}

First, we give a brief high-level description of this common coin primitive. 
Nodes first initialize the random coin by running the AERID primitive from Subsection \ref{subsec:informationDissemination}. More precisely, each node initiates $\Theta(n \log n)$ random walks (or tokens) and each walk is stored separately in a distributed fashion. The remainder of the random coin primitive --- the coin flipping component --- is split into phases. In each phase, some designated senders transmit a random bit using the previously-computed random walks only. This desired behavior is enforced by the honest nodes; the honest nodes ensure that only messages travelling along the random walk paths are retransmitted, whereas other (incorrect) messages are ignored. At the end of the phase, nodes take the majority bit among all received messages' random bits as their local random coin output.

Next, we give some intuition for the common coin primitive's correctness. To start with, reusing the random walks for the coin flipping component ensures that for significantly many phases out of any $n$ successive phases, almost all of the honest nodes receive more good messages, having visited only honest nodes and containing a random bit that originates from a designated sender, than bad messages, originating at or tampered by Byzantine nodes. Moreover, for a constant fraction of these phases, it holds that there is a unique designated sender and thus that almost all honest nodes agree on a random bit.

\paragraph{Detailed Primitive Description.} First, the (eventual almost-everywhere) common coin must be initialized. Once that is done, the primitive can be invoked to produce a coin flip. Nodes keep a counter $i$ of how many (coin flip) calls have been executed until now. From a more technical perspective, the initialization is done by calling an $InitCoin()$ function and the $i$th coin flip by calling a $CoinFlip(i)$ function, for any integer $i \geq 1$.

When initializing, all nodes randomly select an integer (or rank) in $[1,n]$. Then, nodes run the AERID primitive described in Subsection \ref{subsec:informationDissemination} (with the slight modification that tokens contain the initiating node's rank in addition to all the other information). In more detail, recall that in the AERID primitive, each node initializes $\Theta(n \log n)$ (random walk) tokens, each taking $O(\log n)$ steps. Moreover, each token is uniquely identified by the combination of the originator node's ID and rank as well as a token counter (distinguishing different tokens originating at the same node). Each node that is visited by a token stores the incoming and outgoing edge, the token's originator ID, rank and counter, as well as the number of steps taken by the token. Storing this information allows us, after the initialization and in particular during the coin flipping phases, to send messages along specifically chosen random walks --- in particular, along walks taken by messages with a given rank --- and ensure that all other messages (i.e., not following such a pattern) are ignored.

When the $i$th coin flip is invoked (i.e., the $i$th coin flipping phase), the nodes with rank $i \;\mathrm{mod}\ n$ (and only they) are the \emph{designated senders}. They each flip a coin (i.e., pick 0 or 1 with probability $1/2$ each) and then almost-everywhere broadcast these bits, using the previously computed random walks. More precisely, let $r_v$ be the random bit of designated sender node $v$. Then $v$ almost-everywhere broadcasts $r_v$ by generating $\Theta(n \log n)$ different messages --- containing the random bit $r_v$, the rank of $v$ and a unique token counter corresponding to one used by $v$ during the initialization --- and sending these messages along the $\Theta(n \log n)$ random walks starting at $v$. Other (honest) nodes ensure that only messages corresponding to tokens with rank $i$ are transmitted (and know over which edges to forward the message using the information they stored during the initialization). All tokens reach their destination after $\tilde{\Theta}(n)$ rounds --- even though the walks have length $O(\log n)$ --- which is due to congestion (from the $\Theta(n \log n)$ different messages with the same rank). After which, all nodes take among all messages received via random walks, the majority bit as the bit transmitted by this phase's designated sender.

\paragraph{Analysis.} To start with, we show that there exist sufficiently many phases for which there exist a single honest designated sender --- see Lemma \ref{lem:noCollision} below. 

\begin{lemma}
\label{lem:noCollision}
Consider only the honest nodes. For large enough $n$, it holds with high probability that at least $\frac{n}{8}$ ranks are chosen by exactly one honest node.
\end{lemma}

\begin{proof} 
Let us analyze the number of ranks chosen by exactly one honest node. We model this as a balls-and-bins scenario, in which the $H = n-o(n)$ honest nodes each throw one ball into the $n$ slots (i.e., bins), and use the Poisson approximation approach.
%(see Section \ref{sec:prelim}). 

Let $X_1^{(H)},\ldots,X_n^{(H)}$ be the number of balls thrown in the first to $n$th bins under this balls-and-bins scenario. Let $Y_1^{(H)},\ldots,Y_n^{(H)}$ be independent Poisson random variables with mean $\mu_P = H/n = 1-o(1)$, where a Poisson random variable $X$ with parameter $\mu$ is a discrete random variable taking values in $\mathbb{N}$ with distribution $\Pr[X = k] = \frac{\mu^k e^{-\mu}}{k!}$. Let $f(Y_1^{(H)},\ldots,Y_n^{(H)})$ be the number of bins (under the Poisson distribution) with exactly one ball. Then, it is well-known that for any indicator function $f(x_1,\ldots,x_n)$, \[ \Pr[f(X_1^{(m)}, \ldots, X_n^{(m)})] \leq e \sqrt{m} \Pr[f(Y_1^{(m)}, \ldots, Y_n^{(m)})] \]

Now, let $f(Y_1^{(H)},\ldots,Y_n^{(H)}) = \sum_{i=1}^n \I[Y_i^{(H)}=1]$ where $\I[E]$ is the indicator function for event $E$. Note that for any $1 \leq i \leq n$, $\I[Y_i^{(H)}=1]$ is a Bernoulli random variable with parameter $p = \mu_P e^{-\mu_P}$. Hence, $\E[\sum_{i=1}^n \I[Y_i^{(H)}=1]] = n \mu_P \; e^{-\mu_P} = H e^{-\mu_P}$. 
Moreover, using standard Chernoff bounds (see Theorem 4.4 in \cite{MitzenmacherUpfalBook}),
we get $\Pr[f(Y_1^{(H)},\ldots,Y_n^{(H)}) - H e^{-\mu_P}| \geq \frac{1}{2} H e^{-\mu_P}] \leq  2e^{-\frac{H e^{-\mu_P}}{12}}$. Since $H e^{-\mu_P} = (n-o(n)) e^{1-o(1)}$, $H e^{-\mu_P} \geq \frac{n}{2\sqrt{e}} \geq \frac{n}{4}$ for large enough $n$. 
Thus, $f(Y_1^{(H)},\ldots,Y_n^{(H)})$ is greater than $\frac{n}{8}$ with high probability (for large enough $n$).
And as a result, by the above (Poisson approximation approach) inequality, $f(X_1^{(H)},\ldots,X_n^{(H)})$ is greater than $\frac{n}{8}$ with high probability for large enough $n$, and the lemma statement follows.
\end{proof}

Next, we say a phase $i \in [1,n]$ (or coin flip) is \emph{good} if there is exactly one honest designated sender for phase $i$ (i.e., a single honest node chose $i$), and that sender successfully (almost-everywhere) broadcasted during initialization. Then, we show next that there are sufficiently many good phases within any $n$ successive phases, and that any good phase terminates with a random bit (0 or 1 with probability 1/2 each) being shared by almost all nodes.

\begin{lemma}
\label{lem:existsGoodRank}
    At least $n/8 - o(n)$ phases are good with high probability.
\end{lemma}

\begin{proof}
     To start with, Lemma \ref{lem:noCollision} states that for large enough $n$, it holds with high probability that at least $n/8$ ranks are chosen by exactly one honest node. Next, by Theorem \ref{thm:informationDissemination} (and the definition of AERID), it holds with high probability that $n-o(n)$ honest node succeed in almost-everywhere broadcasting during the initialization, or in other words, at most $o(n)$ honest nodes fail in almost-everywhere broadcasting. The lemma statement follows from these two points.
\end{proof}

\begin{lemma}
\label{lem:goodCoinBroadcasts}
  When any good phase $i \geq 1$ terminates, at least $n - o(n)$ honest nodes agree on a common binary value. Moreover, this value is 0 with probability 1/2, and 1 with probability 1/2.
\end{lemma}

\begin{proof}
    Consider some good phase $i \geq 1$. By definition, there is a single honest designated sender $v$, and $v$ chooses a random binary value uniformly at random. Moreover, $v$ must have succeeded in almost-everywhere broadcasting in the initialization part. Now, in phase $i$, $v$ transmits its random bit via the random walks computed during initialization. Recall that messages corresponding to these random walks are transmitted by honest nodes, whereas messages that do not correspond to these walks are ignored by honest nodes, which implies that $v$ also almost-everywhere broadcasts in phase $i$. More concretely, by Theorem \ref{thm:AEtoAEBroadcast} at least $n-o(n)$ honest nodes receive $\Theta(\log n)$ good tokens (that only visit honest nodes in the core and thus contain the random bit of $v$) and $o(\log n)$ bad tokens (having possibly visited Byzantine nodes and that may have the opposite random bit). Thus, these $n-o(n)$ honest nodes obtain $v$'s random bit after taking the majority bit out of its received messages. The lemma statement follows.
\end{proof}

\begin{lemma}
\label{lem:randomCoin}
    The proposed primitive correctly implements an eventual almost-everywhere common coin (EAECC) tolerating up to $|B| = o(n / \log n)$ Byzantine nodes. The initialization takes $\tilde{O}(n)$ rounds and each coin flip takes $\tilde{O}(n)$ rounds.
\end{lemma}

\begin{proof}
    The correctness follows from Lemmas \ref{lem:existsGoodRank} and \ref{lem:goodCoinBroadcasts}.
    As for the runtime, that of the initialization phase follows from Theorem \ref{thm:informationDissemination} whereas that of the coin flips follows from the $\tilde{O}(n)$ congestion incurred by the messages ($\tilde{O}(n)$ of them per phase) travelling along the $O(\log n)$ length random walks.
\end{proof}

\subsection{Faster Eventual Almost-Everywhere Common Coin}
\label{subsec:fasterCoin}

This second eventual almost-everywhere common coin (EAECC) primitive follows the same schema as that of the previous subsection. The main difference lies in our use of the second AERID primitive (see Subsection \ref{subsec:fasterCoin}) for the initialization. Although the paths computed during the initialization are more complex, they have nicer congestion properties and we use this to speed up the coin flips compared to the previous subsection. 

In more detail, the coin is initialized by executing the AERID primitive described in Subsection \ref{subsec:fasterCoin}. This primitive also generates $\Theta(n \log n)$ tokens per node, but in a more gradual and indirect fashion. Indeed, each node does not directly generate these tokens, but only sends $\tilde{\Theta}(\log^{3+2\varepsilon} n)$ random walks (or tokens) (for any arbitrarily chosen small $\varepsilon > 0$). The tokens generated by some node $v$ are said to have the source ID of $v$. They take $\Theta(\log n)$ steps, before getting duplicated twice; this possibly, or in fact likely, happens at some other node $u$ but without changing the source ID (which remains that of $v$). The process is repeated until we have $\Theta(n \log n)$ tokens with the source ID of $v$. The fact that these tokens get duplicated (mostly) at nodes beside $v$ greatly helps out in reducing the congestion when the token's paths are reused for the coin flips. In fact, we will show that the edge congestion is $\tilde{O}(1)$, which should be compared with the $\tilde{\Theta}(n)$ edge congestion within the coin flips in the previous subsection.
Note that here again, we use the AERID primitive with a slight modification: tokens contain the source ID's rank, in addition to the source ID and a token counter. The latter can be a pair where the first element is a stage number and the second element a unique ID for all tokens with the same source duplicated at a given node when that stage starts, such that the pair uniquely determines all tokens with a given source ID. Each node that is visited by a token stores the incoming and outgoing edge, the token's source ID, rank and counter as well as both the stage and steps during which the visit happens. Once again, storing this information allows us to send, during the coin flipping phases, messages along specifically chosen random walks (e.g., corresponding to a given rank) while ensuring messages that diverge from the stored walk are ignored. 

Next, we detail how, within the $i$th coin flip, a designated sender $v$ almost-everywhere broadcasts its randomly chosen bit $r_v$. This is done as follows. Node $v$ generates, for each of the $\Theta(\log^{3+2\varepsilon} n)$ tokens of the initialization, exactly as many (i.e., $\Theta(\log^{3+2\varepsilon} n)$) messages that contain the random bit $r_v$, the rank of $v$ and the corresponding (unique) token counter (which, as mentioned previously, is a pair here). These messages walk along the $\Theta(\log n)$ steps taken by the corresponding token. Each such step takes some $O(\log^{4+3\varepsilon} n)$ rounds and any message that cannot be transmitted because of edge overcongestion is simply discarded. (However, the AERID primitive we use ensures that for most ranks, no message fails to be transmitted due to edge overcongestion.) At the message's destination (i.e., the destination of the corresponding token during the first stage of the AERID primitive), the message is duplicated as many times as the corresponding token was, and the associated information (in particular the token counter) is updated to match the initialization phase. This repeats until the messages reach the last node reached by their corresponding tokens (at the end of the AERID primitive). Throughout this process, honest nodes ensure only messages with the rank $i$ are transmitted (and know how to forward the messages along the walks due to the information stored during the initialization). Crucially, the properties of the second AERID primitive ensures that the almost-everywhere broadcast fails for $o(n)$ ranks only, whether it is because too many tokens visited a Byzantine node during the initialization (and thus the random bit may get corrupted during the coin flip) or because too many messages take the same edge for a given step.

\paragraph{Analysis.} Since ranks are chosen here exactly as in the previous EAECC primitive, Lemma \ref{lem:noCollision} also applies here. Recall that we say that a phase $i \in [1,n]$ (or coin flip) is \emph{good} if there is exactly one honest designated sender for phase $i$ (i.e., a single honest node chose $i$) and that sender successfully (almost-everywhere) broadcasted during initialization. Here, we say that a good phase $i$ is additionally \emph{low-congestion} if the sender successfully (almost-everywhere) broadcasted during initialization without incurring high edge congestion, or more concretely, if even if it successfully (almost-everywhere) broadcasted even if good tokens except the low-congestion good tokens are discarded. Then, we show that there are sufficiently many good and low-congestion phases within any $n$ successive phases.

\begin{lemma}
\label{lem:existsGoodRank2}
    At least $n/8 - o(n)$ phases are good and low-congestion with high probability.
\end{lemma}

\begin{proof}
     To start with, Lemma \ref{lem:noCollision} states that for large enough $n$, it holds with high probability that at least $n/8$ ranks are chosen by exactly one honest node. Next, by Theorem \ref{thm:informationDisseminationCongestionFree} (and the definition of AERID), it holds with high probability that $n-o(n)$ honest node succeed in almost-everywhere broadcasting during the initialization even if good tokens except the low-congestion good tokens are discarded. Or in other words, at most $o(n)$ honest nodes fail in almost-everywhere broadcasting even if when the coin flip is executed, all walks are allowed only $O(\log^4 n)$ rounds per step. The lemma statement follows from these two points.
\end{proof}

\begin{lemma}
\label{lem:goodCoinBroadcasts2}
  When any good and low-congestion phase $i \geq 1$ terminates, at least $n - o(n)$ honest nodes agree on a common binary value. Moreover, this value is 0 with probability 1/2, and 1 with probability 1/2.
\end{lemma}

\begin{proof}
    Consider some good phase $i \geq 1$. By definition, there is a single honest designated sender $v$, and $v$ chooses a random binary value uniformly at random. Moreover, $v$ must have succeeded in almost-everywhere broadcasting in the initialization part, even if good tokens except the low-congestion good tokens are discarded. Now, in phase $i$, $v$ transmits its random bit via the random walks computed during initialization. Recall that messages corresponding to these random walks are transmitted by honest nodes, whereas messages that do not correspond to these walks are ignored by honest nodes. Moreover, at most $O(\log^{4+3\varepsilon} n)$ rounds are allowed for each step and stage of these random walks, which allows at most $O(\log^{4+3\varepsilon} n)$ messages to transit through this edge in that step and that stage. Since low-congestion tokens (see Subsection \ref{subsec:congestionAERID}) by definition transit through edges with at most $O(\log^{4+3\varepsilon} n)$ congestion per step and stage, this implies that $v$ also almost-everywhere broadcasts in phase $i$ (despite the fact that some messages may be discarded due to runtime limitations). More concretely, by Theorem \ref{thm:correctnessCongestionAERID} at least $n-o(n)$ honest nodes receive $\Theta(\log n)$ good and low-congestion tokens (that only visit honest nodes in the core and thus contain the random bit of $v$, and never transit through an over-congested edge) and $o(\log n)$ bad tokens (having possibly visited Byzantine nodes and that may have the opposite random bit) with the source ID of $v$ during the initialization. Thus, during the $i$th coin flip these $n-o(n)$ honest nodes obtain $v$'s random bit after taking the majority bit out of its received messages. The lemma statement follows.
\end{proof}

\begin{lemma}
\label{lem:randomCoin2}
    The proposed primitive correctly implements an eventual almost-everywhere common coin (EAECC) tolerating up to $|B| = o(n / \log^2 n)$ Byzantine nodes. The initialization takes $\tilde{O}(n)$ rounds and each coin flip takes $\tilde{O}(1)$ rounds.
\end{lemma}

\begin{proof}
    The correctness follows from Lemmas \ref{lem:existsGoodRank2} and \ref{lem:goodCoinBroadcasts2}.
    As for the runtime, that of the initialization follows from Theorem \ref{thm:informationDisseminationCongestionFree} whereas that of the coin flips follows from the fact that during a coin flip, each step is allowed $\tilde{O}(1)$ rounds only, and messages take at most $O(\log^2 n)$ steps (as there are $O(\log n)$ steps per stage in the AERID primitive, and $O(\log n)$ stages), so the coin flip takes $\tilde{O}(1)$ rounds.
\end{proof}   

\section{Fully-Distributed Byzantine Agreement Protocol}
\label{sec:BASparse}
In this section, we present our main result: a fully-distributed almost-everywhere Byzantine agreement algorithm (Algorithm \ref{alg:ByzantineAgreement}). This algorithm is based on Rabin's algorithm \cite{Rabin_1983} and the main difficulty lies in implementing a random common coin primitive in a fully-distributed fashion. Due to the fully-distributed constraint, we settle on implementing an eventual almost-everywhere common coin (or EAECC, see Section \ref{sec:randomCoin}), and we show that this suffices to solve almost-everywhere Byzantine agreement but with a significantly slower runtime (compared to the setting of Rabin's algorithm \cite{Rabin_1983}, in which the common coin is provided to all nodes by a trusted third party).

First, we give a high-level description of Algorithm \ref{alg:ByzantineAgreement}. The algorithm runs for $p = \Theta(n \log n)$ phases (of either $\tilde{O}(n)$ or $\tilde{O}(1)$ rounds, depending on the EAECC protocol used) and terminates afterward. Nodes start with their vote set to their input value. In each phase, nodes check if there already exists a strong majority of nodes (i.e., 0.9 of them) that agree on some vote. If so, nodes change their vote accordingly, to that majority vote. Otherwise, nodes flip the eventual almost-everywhere common coin and set their vote to the coin's output. Note that while Rabin's algorithm terminates in an expected constant number of phases, or terminates with high probability within $O(\log n)$ phases, Algorithm \ref{alg:ByzantineAgreement} requires significantly more phases because we use a weaker \emph{eventual} common coin.

Next, we give a more precise description of Algorithm \ref{alg:ByzantineAgreement}.
The algorithm starts by initializing the eventual almost-everywhere common coin (using $InitCoin()$, see Section \ref{sec:randomCoin}).
After which, the algorithm runs $p$ phases, each decomposed into two subphases. Consider phase $i \in [1,p]$. In the first subphase, all nodes run the Byzantine Random Walk protocol (Algorithm \ref{alg:byzantineSamplingSparse}) from Section \ref{sec:byzantineRandomWalk} for $\tilde{O}(1)$ rounds. More concretely, each node initiates some $\tot = \Theta(\log^3 n)$ tokens that contain the node's vote and execute a random walk of $O(\log n)$ length. Doing so guarantees that almost all nodes sample with good precision the proportion of both votes, and thus allows almost all nodes to detect if there exists a strong majority.
In the second subphase, nodes flip the EAECC (by calling $CoinFlip(i)$, see Section \ref{sec:randomCoin}). If the first EAECC primitive is used, then this takes $\tilde{O}(n)$ rounds but tolerates up to $|B| = o(n/\log n)$ Byzantine nodes. Whereas if the second EAECC primitive is used, this takes $\tilde{O}(1)$ rounds but tolerates up up to $|B| = o(n/\log^2 n)$ Byzantine nodes only. 
Once the two subphases are done, each node $v$ first checks if it detected a strong majority in the first subphase.
If so, $v$ changes its vote (if different) to that majority vote. Otherwise, $v$ changes its vote to the coin's output obtained during the second subphase.
Finally, once all $p$ phases are done, nodes terminate with their current vote.

\begin{algorithm}[ht]
\caption{Byzantine Agreement Algorithm for honest node $v$ with input $b_v$}
\label{alg:ByzantineAgreement}
\begin{algorithmic}[1]

\State $vote_v := b_v$
\State $InitCoin()$  

\For{phase $i = 1$ to $p$} 
    \State $samples_v$ are the tokens obtained by running Algorithm \ref{alg:byzantineSamplingSparse} with $\tot = \Theta(\log^3 n)$
    \State $maj_v :=$ the majority vote among the votes in $samples_v$
    \State $tally_v :=$ the number of majority votes in $samples_v$ divided by $n$
    \State $bit_v := CoinFlip(i)$

    \If{$tally_v > 0.9$} 
        \State $vote_v := maj_v$ \Comment{Set to majority vote} \label{line:maj}
    \Else
        \State $vote_v := bit_v$  \Comment{Set to EAECC flip} \label{line:coin}
    \EndIf

\EndFor

\end{algorithmic}
\end{algorithm}

Now, we analyze the behavior of Algorithm \ref{alg:ByzantineAgreement}. We remind that good tokens are tokens that only ever visit good nodes (and in particular the core subgraph) and thus remain uncorrupted. Other tokens, which we call bad, may be corrupted (i.e., their vote changed) when visiting a Byzantine node.

\begin{lemma}
\label{lem:estimationByz}
    Consider some phase $i \in [1,p]$. During the sampling subphase, each (honest) node sends out $\tot = \Theta(\log^3 n)$ tokens via the Byzantine Random Walk protocol. Let $f$ be the fraction of some vote $m$ held by honest nodes. Then there exists a large enough subset $R \subseteq C$ of honest nodes of size $|R| = n - o(n)$ such that for any node $u \in R$, node $u$ receives $(f \pm o(1)) \tot$ good tokens with vote $m$ and $o(\tot)$ bad tokens.
\end{lemma}

\begin{proof}
    The sampling phase executes the Byzantine Random Walk protocol with $\tot = \Theta(\log^3 n)$. Recall that $f$ is the fraction of the vote $m$ held by honest nodes. This implies that there are $(f \pm o(1)) |C|$ core nodes with vote $m$, and these core nodes generate $(f \pm o(1)) |C| \tot$ tokens containing the vote $m$.  
    
    We show that most of these tokens disseminate uniformly throughout the core and obey concentration bounds. First, by Items (1) and (2) of Theorem \ref{thm:ByzSamplingFinal}, at least $(f \pm o(1)) |C| \tot - o(|C| \tot) = (f \pm o(1)) |C| \tot$ of these tokens are good, i.e., they walk only in $C$ (for $\kappa = o(1)$). Hence, each such walk mixes in the core, and ends at some node $u \in C$ with probability $p(u) = deg_C(u)/(2|E_C|) \pm 1/n^3$, by Item (3) of Theorem \ref{thm:ByzSamplingFinal}. By Lemma \ref{lem:boundCoreEdges} where $\mu = |B|/|C| = o(1)$, $(1-o(1)) d|C|/2 \leq |E_C| \leq d |C|/2$ (for $\mu = |B|/|C| = o(1)$). As a result, $1/|C| - 1/n^3 \leq p(u) \leq (1+o(1))/ |C| + 1/n^3$.
    Therefore, any node $u \in C$ receives in expectation $(f \pm o(1)) \tot$ good tokens containing vote $m$. Since the good tokens execute independent random walks, and $\tot = \Theta(\log^3 n)$, we can use standard Chernoff bounds (see Theorem 4.4 in \cite{MitzenmacherUpfalBook}) to prove that any node $u \in C$ receives $(f \pm o(1)) \tot$ good tokens containing vote $m$ whp.
    
    Finally, by Item (1) of Theorem \ref{thm:ByzSamplingFinal}, at most $o(|C| \tot)$ bad tokens enter and end in $C$ (for $\kappa = o(1)$). This implies that there exists a large enough subset $C' \subset C$ with $|C'| = |C| - o(|C|)$, such that any node $u \in C'$ ends with $o(\tot)$ bad tokens. Thus, we get the lemma statement. 
\end{proof}

\begin{lemma}
\label{lem:maintainingAgreement}
Consider some phase $i \in [1,p]$. If at least an $f = 1-o(1)$ fraction of the honest nodes start the phase with the same vote, then with high probability at least a $1-o(1)$ fraction of the honest nodes end the phase with that vote.
\end{lemma}

\begin{proof}
At least $n - o(n)$ honest nodes start the phase with the same vote, denoted by $m$, which implies that vote $m$ is held by a fraction $f = 1-o(1)$ of honest nodes. Thus, by Lemma \ref{lem:estimationByz}, with high probability there exists a subset $R$ of $n-o(n)$ nodes that receive $(1-o(1)) \tot$ good tokens with vote $m$, $o(\tot)$ good tokens with the other vote and $o(\tot)$ bad tokens. By the algorithm description, all nodes in $R$ detect that $m$ is in a strong majority and thus end the phase with vote $m$.
\end{proof}

\begin{lemma}
\label{lem:reachAgreement}
Consider some phase $i \in [1,p]$. If the $i$th coin flip (phase) is good, then almost-everywhere agreement is reached at the end of the phase with probability $1/2 - o(1)$. 
\end{lemma}

\begin{proof}
First, note that either (a) at least $n-o(n)$ nodes do not detect a strong majority (i.e., have both tallies smaller than $0.9$), or (b) $\Omega(n)$ nodes detect a strong majority (i.e., have one of the two tallies strictly greater than $0.9$). For the simpler case (a), at least $n-o(n)$ nodes set their vote to the coin's value by the description of Algorithm \ref{alg:ByzantineAgreement}, thus reaching (almost-everywhere) agreement. 

As for case (b), it implies that out of the two votes, $\Omega(n)$ honest nodes consider one of the two, say $m$ without loss of generality, to be a strong majority. Note that this does not rule out, for now, $\Omega(n)$ other honest nodes considering the other vote to be a strong majority as well. However, we next prove that in fact, at least $n-o(n)$ honest nodes agree on $m$ as the majority vote. Indeed, by Lemma \ref{lem:estimationByz}, at least $0.9-o(1)$ honest nodes started the phase with vote $m$. However, this in turn implies that $m$ starts the phase as a majority vote held by at least $0.9-o(1)$ honest nodes. Thus, by Lemma \ref{lem:estimationByz}, at least $n-o(n)$ honest nodes agree on $m$ as the majority vote (but not necessarily as a strong majority). 

We next show that with probability $1/2$, the honest nodes that agree on $m$ as a majority vote but did not detect a strong majority set their vote to $m$ by the end of the phase. Indeed, Lemma \ref{lem:goodCoinBroadcasts} implies that during the second subphase, at least $n - o(n)$ honest nodes agree on a common value $b$, which is 0 with probability 1/2 and 1 with probability 1/2. Moreover, this value is chosen independently of $maj$, as $maj$ is fixed (possibly influenced by the Byzantine adversary) by the end of the first subphase of phase $i$ whereas the coin flip happens later, in the second subphase. (Thus even a full information Byzantine adversary cannot deduce the output of that coin flip during the first subphase.) Hence, $\Pr[b = maj] = 1/2$. As a result, at least $n-o(n)$ honest nodes set their vote to $b$ with probability $1/2-o(1)$. The lemma statement follows.
\end{proof}

\begin{theorem}
\label{thm:mainResult}
Let $G$ be an expander graph having $n$ nodes out of which a subset of $|B| = o(n / \log n)$ nodes are Byzantine. Then, there exists a fully-distributed algorithm solving almost-everywhere Byzantine agreement (AEBA) with high probability. Moreover, it does so in $\tilde{O}(n^2)$ rounds.
\end{theorem}

\begin{proof}
We start by the correctness. First, note that once almost-everywhere agreement is reached in some phase $i \in [1,p]$, then nodes maintain almost-everywhere agreement for all subsequent phases $i' > i$ (by Lemma \ref{lem:maintainingAgreement}). Hence, it suffices to show that nodes reach almost-everywhere agreement in at least one phase with high probability. By Lemma \ref{lem:existsGoodRank}, within any successive $n$ coin flips, at least one coin flip (phase) is good. Thus within the $p = \Theta(n \log n)$ phases, at least $\Omega(\log n)$ of the coin flips are good. Moreover, independently for each such phase, almost-everywhere agreement is reached at the end of the phase with probability at least $1/2 - o(1)$ (by Lemma \ref{lem:reachAgreement}). Hence, nodes reach almost-everywhere agreement in at least one phase with high probability and Algorithm \ref{alg:ByzantineAgreement} solves almost-everywhere Byzantine agreement with high probability.

Next, we consider the round complexity. First, setting up the coin primitive takes $\tilde{O}(n)$ rounds by Lemma \ref{lem:randomCoin}. As for the $p$ phases, the sampling subphase takes $\tilde{O}(1)$ rounds (by Theorem \ref{thm:ByzSamplingFinal}) whereas the coin flip subphase takes $\tilde{O}(n)$ rounds (by Lemma \ref{lem:randomCoin}). Hence, Algorithm \ref{alg:ByzantineAgreement} takes $\tilde{O}(n^2)$ rounds.
\end{proof}

\begin{theorem}
\label{thm:mainResult2}
Let $G$ be an expander graph having $n$ nodes out of which a subset of $|B| = o(n/\log^2 n)$ nodes are Byzantine. Then, there exists a fully-distributed algorithm solving almost-everywhere Byzantine agreement (AEBA) with high probability. Moreover, it does so in $\tilde{O}(n)$ rounds.
\end{theorem}

\begin{proof}
    The correctness and runtime can be shown following the proof of Theorem \ref{thm:mainResult}. For the latter, note that the coin flip subphase takes $\tilde{O}(1)$ rounds (by Lemma \ref{lem:randomCoin2}) and thus Algorithm \ref{alg:ByzantineAgreement} takes $\tilde{O}(n)$ rounds only.
\end{proof}

\section{Conclusion and Open Problems}
\label{sec:conclusion}

We address the fundamental Byzantine agreement problem in sparse (bounded-degree) networks, a practically relevant setting
to real-world networks, especially modern P2P networks that underlie blockchains and cryptocurrencies.
In these networks, it is crucial to develop efficient fully-distributed protocols which operate with only local (initial) knowledge.
In this work, we develop  fully-distributed protocols that tolerate a large number of Byzantine nodes --- up to $o(n/\log n)$. This answers open questions raised in previous works \cite{Dwork_1988, King_2006_FOCS} of whether such algorithms are possible.

Several key questions remain. Our protocols run in a polynomial number of rounds. In particular, one of our protocols
runs in near-linear $\tilde{O}(n)$ rounds while tolerating $o(n/\log^2 n)$ Byzantine nodes. It is not clear whether this is the best possible round complexity for tolerating a nearly linear number of Byzantine nodes or whether significantly faster (say, $\polylog{n}$ round) algorithms are possible.  Unlike complete networks where there are well-established message  lower bounds\footnote{Note that there are fast $O(\log n)$-round algorithms for BA in complete networks that tolerate even up to nearly $n/3$ Byzantine nodes, but these take at least quadratic messages \cite{Ben-Or_2006,Goldwasser_2006}.} (e.g., \cite{DolevR85,HadzilacosH93}),
we are not aware of message or time lower bounds on the runtime of Almost-Everywhere Byzantine Agreement (AEBA)
in {\em sparse} networks.  Since we require only almost-everywhere agreement and  Byzantine nodes can only communicate through
the graph edges, the power of the adversary is somewhat reduced  (compared to complete networks), and it is not clear how to show lower bounds in sparse networks. In particular, there is a striking contrast between the two settings. In complete networks, it is known that $\Omega(nt)$ messages are necessary even for randomized algorithms \cite{HadzilacosH93}. However, in sparse networks, for $t = \sqrt{n}/\polylog{n}$, there is a $O(\log^3 n)$ round algorithm~\cite{Augustine_2013_PODC}. Is it the case
that the $\Omega(nt)$ message lower bound holds in the sparse setting for higher values of $t$, in particular when $t$ is near-linear in $n$? If so, then since the degree is bounded, the $\Omega(nt)$ message lower bound will imply that
$\Omega(t)$ is a lower bound on the round complexity of AEBA protocols (under bandwidth constraint)  that tolerate up to $t$ Byzantine nodes.
If this is true, then our second protocol will be nearly-optimal.

\newpage

\bibliographystyle{plainurl}
\bibliography{reference}

\end{document}